\documentclass[journal]{IEEEtran}
\usepackage{times}

\usepackage[numbers]{natbib}
\usepackage{multicol}
\usepackage{hyperref}
\IEEEoverridecommandlockouts                              
\usepackage{makeidx}         
\usepackage{graphicx}        
\usepackage{epstopdf}
\usepackage{multicol}        

\usepackage{amsmath,amssymb}
\usepackage{euscript}
\usepackage{subfig}
\usepackage{url}
\usepackage{mathrsfs}
\usepackage{array}
\usepackage{wrapfig}

\usepackage{epsfig,float,alltt}
\usepackage{psfrag,xr}
\usepackage{pdfpages}
\usepackage{epstopdf}
\usepackage{pstool}

\usepackage[T1]{fontenc}

\usepackage{multirow}

\usepackage{balance} 
\usepackage{color} 

\newtheorem{theorem}{\bfseries Theorem}

\newtheorem{lemma}{\bfseries Lemma}

\newtheorem{remark}{\bfseries Remark}

\newcommand{\gap}{\vspace{0.2cm}}
\newcommand{\HRule}{\noindent\rule{\linewidth}{0.1mm}\newline}

\newcommand{\argmin}{\mathop{\mathrm{argmin}}}

\renewcommand{\epsilon}{\varepsilon}

\newcommand{\etavar}{\eta}

\newcommand{\Orbit}{\mathscr{O}}
\newcommand{\Orbitz}{\mathscr{O}_Z}
\newcommand{\HS}{\mathscr{H}}

\renewcommand{\S}{S}

\newcommand{\Rx}{\Delta_X}

\newcommand{\Ve}{V_{\epsilon}}

\newcommand{\TI}{T_{I}^{\epsilon}(\eta,z)}

\newcommand{\mytag}[1]{\tag*{(\textbf{#1})}}

\newcommand{\update}[1]{\textcolor{black}{#1}}
\newcommand{\editQTN}[1]{\textcolor{black}{#1}}
\newcommand{\updateT}[1]{\textcolor{black}{#1}}

\begin{document}

\title{Robust Safety-Critical Control for Dynamic Robotics}
\author{Quan Nguyen and Koushil Sreenath
\thanks{Q. Nguyen is with the Department of Aerospace and Mechanical Engineering, University of Southern California, Los Angeles, CA 90007, email: {\tt quann@usc.edu}.}%
\thanks{K. Sreenath is with the Department of Mechanical Engineering, University of California, Berkeley, CA 94720, email: {\tt koushils@berkeley.edu}.}
\thanks{The work of Q. Nguyen is supported by USC Viterbi School of Engineering startup
	funds. The work of K. Sreenath was partially supported through National Science Foundation Grant IIS-1834557.}
}

\maketitle


\begin{abstract}
We present a novel method of optimal robust control through quadratic programs that offers tracking stability while subject to input and state-based constraints as well as safety-critical constraints for nonlinear dynamical robotic systems in the presence of model uncertainty. The proposed method formulates robust control Lyapunov and barrier functions to provide guarantees of stability and safety in the presence of model uncertainty.  We evaluate our proposed control design on dynamic walking of a five-link planar bipedal robot subject to contact force constraints as well as safety-critical precise foot placements on stepping stones, all while subject to model uncertainty. We conduct preliminary experimental validation of the proposed controller on a rectilinear spring-cart system under different types of model uncertainty and perturbations. 
\end{abstract}
\section{Introduction}

Designing controllers for complex robotic systems with nonlinear and hybrid dynamics for achieving stable high-speed tracking of reference trajectories while simultaneously guaranteeing input, state, and safety-critical constraints is challenging.  Constraints on robotic systems arise as limits of the physical hardware (such as work-space constraints, joint position and velocity constraints, and motor torque constraints) as well as constraints imposed by controllers for safe operation of the system (such as collision constraints, range constraints, connectivity constraints, contact force constraints, etc.)  Further, adding to the challenges of stable constrained control is the presence of high-levels of uncertainty in the dynamical model of the robot\updateT{.} 
The goal of this paper is to address the problem of designing stabilizing controllers for systems with strict constraints 
\updateT{in the presence of} model uncertainty.

\subsection{Background}
Lyapunov functions and control Lyapunov functions (CLFs) are a classical tool for design and analysis of feedback control that stabilize the closed-loop dynamics of both linear and nonlinear dynamical systems, see \cite{FK:Book}.
Traditional CLF-based controllers involve closed-form control expressions such as the min-norm and the Sontag controllers \cite{Sontag_Ctrl}.  Recently, a novel approach of expressing CLF-based controllers via online quadratic programs (QPs) in \cite{TorqueSaturation:GaSrAmGr:TRO14} opened an effective way for dealing with stability while also enabling the incorporation of additional constraints, such as input constraints. In \cite{TorqueSaturation:GaSrAmGr:TRO14}, the CLF-based controller is expressed in the QP as an inequality constraint on the time-derivative of the CLF, which easily enables adding additional constraints such as input saturation through the relaxation of the CLF inequality condition. The QP can be executed online at 1kHz in real-time as a state-dependent feedback controller. However, this work did not handle model uncertainty or safety-critical constraints.

In addition to CLFs, we also draw inspiration from recent methods of control barrier functions (CBFs) that can be incorporated with control Lyapunov function based quadratic programs to result in the CBF-CLF-QP, as introduced in \cite{CBF:CBFQPwithACC:CDC14}.  This framework enables handling safety-critical constraints effectively in real-time. Experimental validation of this type of controller for the problem of Adaptive Cruise Control was presented in \cite{CBF:Ames:ExperimentalACC:ACC15}. This framework has also been extended to various interesting application domains, such as safety-critical geometric control for quadrotor systems \cite{CBF:Guofan:SafetyCriteria:ACC15} and safety-critical dynamic walking for bipedal robots \cite{CBF:Quan:Footstep:ADHS15}, \cite{CBF:Ames:CBFBiped:ACC15}. Although this work can handle safety-critical constraints, however a precise model of the system is required to enforce the constraints.
%

Moreover, as presented in \cite{CBF:Ames:RobustnessCBF:ADHS15}, preliminary robustness analysis of the CBFs indicate that the safety-critical constraint will be violated in the presence of model uncertainty, with the amount of violation being bounded by the value of the upper bound of the model uncertainty.  In particular, model uncertainty leads to constraint violation of the safety-critical constraints.
%
%
Recent efforts to improve the robustness of CBFs using \updateT{learning-based approaches such Gaussian process regression have been introduced in \cite{fan2019bayesian,R2pmlr-v120-khojasteh20a,R3ChOrMuBu2019,cheng2020safe,R5pmlr-v120-taylor20a} which use data to estimate the uncertainty.  However, this typically requires large amounts of quality data for high-dimensional systems.  Moreover, we note that in \cite{fan2019bayesian,R3ChOrMuBu2019,cheng2020safe}, the authors assume uncertainty only in the drift vector field term of a control affine system and assume the input vector field to be known precisely.  Nonetheless, with learning-based approaches, better and more accurate state-dependent bounds for the uncertainty can be learnt which could result in less conservative approaches than what is typically possible through robust control. However, on the other hand, such learning-based approaches need to assume various distributions to be Gaussian and do not computationally scale well with state-dimension.}
In this paper, we seek a method to 
simultaneously handle robust stability, robust input-based constraints and robust state-dependent constraint in the presence of model uncertainty.  We will do this through robust control formulations of the CLF, CBF, and constraints.The framework can be easily achieved by extending the nominal CBF-CLF-QP controller while improving significantly the robustness of the entire approach with a provable stability and safety guarantee.

\editQTN{QP based controllers for robotic systems have been increasingly used in recent years. Especially, the DARPA Robotics Challenge inspired several new methods in this direction \cite{feng2015optimization}, \cite{step:footstep:Russ:Humanoid14}, wherein QP controllers are used to formulate inverse kinematic and inverse dynamic control problems while minimizing tracking of desired accelerations. While these QP controllers enforce tracking performance and constraints at each time step, our method uses CLF and CBF condition to guarantee stability and safety constraints for future time.}

\editQTN{For control of constrained systems, Model Predictive Control (MPC) has been widely used in many industrial applications \cite{qin2003MPCsurvey}. However, the method is computationally expensive. \updateT{Firstly, for nonlinear systems, MPC typically linearizes the model and the constraints.  Next, e}ven with a linear model and polyhedral constraints, the traditional approach can require a solving time from \updateT{tens of milliseconds and higher even with recent advances} \cite{wang2010fastMPC}\updateT{. For nonlinear control affine systems, CBF constraints are linear inequalities on the input and can be quickly enforced through quadratic programs that are solved point-wise in time in under 1 ms.}}

Robust control is an extensively studied topic.
We have established methods, such as $H_\infty$-based robust control and linear quadratic Guassian (LQG) based robust control \cite{book:OptimalControl, book:RobustAndOptimalControl:Doyle:96} for robust control of linear systems. For robust control of nonlinear systems, input-to-state stability (ISS) and sliding mode control (SMC) are two main methods. The ISS technique (see \cite{iss:Sontag:EJC95, iss:Sontag:SCL95, iss:Sontag:integral:TAC2000}) can be used to both analyze the robustness of nonlinear systems as well as design robust controllers based on control Lyapunov functions. However, a primary disadvantage of ISS based controllers is that the resulting controller only maintains the system errors in a sufficiently small neighborhood of the origin, resulting in non-zero tracking errors. In recent years, there has been work on robust control of hybrid systems based on the ISS technique, for instance see \cite{iss:switched:Liberzon:Automatica07, iss:impulsive:Liberzon:Automatica08, iss:hybrid:ARTeel:SCL09}. 
In contrast, sliding mode control techniques can deal with a wide range of uncertainties and drive the system to follow a given sliding surface, thereby driving outputs to desired values without any tracking errors (see \cite{smc:book:TheoryAndApplications:98, smc:DynamicSystems:IJC92, su1993sliding, fahimi2007sliding}).  However, the primary disadvantage of SMC is the chattering phenomenon caused by discrete switching for keeping the system attracted to the sliding surface.

\update{Robust control techniques have also been extensively applied to robotic manipulator arms, see \cite{Spong1992,LiBr1998}, however manipulator arms do not have challenges such as underactuation and unilateral ground contact forces making their control easier.  Robust control techniques have also been applied to bipedal walking.  For instance, the work in \cite{LiGe2013} extends adaptive robust control of manipulator arms to bipedal robots, the work in \cite{Ruina_Robust2dWalking} considers a simple 2D inverted-pendulum model and pre-computes a control policy through offline nonlinear optimization to prevent falls under the assumption of bounded disturbances.
Similarly, recent results on robust feedback motion planning for the problem of UAVs avoiding obstacles in \cite{Tedrake_RobustMotionPlanning} also precomputed a library of ``funnels'' via convex optimization that represents different maneuvers of the system under bounded disturbances.  A real-time planner then composes motion trajectories based on the resulting funnel library. 
\editQTN{Recent work in reachability-based analysis has been extended to enforce safety-constraints with model uncertainty by safe trajectory synthesis \cite{KoVaJRVa2017,FKHeFiDeTo2017}.}
Offline optimization for stabilization of walking and running with robustness to discrete-time uncertainties such as terrain perturbations has been carried out in \cite{HaGr2013,HaGr2015}.
Our method, in contrast, is based on real-time feedback controller to guarantee robust stability as well as robust safety of robotic systems through an online optimization without the need of precomputed motion plans.
Additional robust planning and control techniques exist for legged robots where the robustness is with respect to stochastic uncertainty in the model / terrain, for instance see \cite{KaPoSuTa2015,SaBl2015,MeLa2013,ByTe2009}.
}

\subsection{Contributions}
The main contributions of this paper with respect to prior work are as follows:
\begin{itemize}
	\item Introduction of a new technique that optimally introduces robust control via quadratic programs to handle stability, input-based constraints and state-dependent constraints under high levels of model uncertainty. 
	\item Robust stability and robust safety-critical constraints achieved through a min-max inequality constraint on the time-derivative of a control Lyapunov function and control barrier function respectively.
	\item Theoretical stability analysis for the QP controllers with relaxed CLF inequality.
	\item Numerical validation of the proposed controller on different problems: 
	\begin{itemize}
		\item Dynamic walking of a bipedal robot while carrying an unknown load and subject to contact force constraints; and
		\item Dynamic walking of a bipedal robot while carrying an unknown load, subject to contact force constraints and precise foot-step location constraints.
	\end{itemize}
	\item Experimental validation of the proposed control method on a rectilinear spring-cart system.
\end{itemize} 

Note that preliminary results of this work were presented in \cite{NgSr2014Robust_RSS,CBF:Quan:RobustCBF:ACC16}.  In contrast to the preliminary results, this paper presents (a) an entirely new min-max formulation of the proposed controller, (b) additional numerical examples for control of a bipedal robot, (c) experimental validations of the proposed controller on a rectilinear spring-cart system and (d) detailed stability and safety analysis for QP controllers with relaxed CLF inequality and the proposed robust CBF-CLF-QP controller.
\subsection{Organization}
The rest of the paper is organized as follows.  Section \ref{sec:CLF_CBF_revisited} revisits control barrier functions and control Lyapunov functions based quadratic programs (CBF-CLF-QPs). Section \ref{sec:proof_main_result} presents the stability analysis for the QP controllers with relaxed CLF inequality. Section \ref{sec:robust_control} discusses the adverse effects of uncertainty and then presents the proposed optimal robust control and formulates as a QP.  Section \ref{sec:simulation} presents numerical and experimental validation on different dynamical robotic systems. Section \ref{sec:conclusion} provides concluding remarks.

\section{Control Lyapunov Functions and Control Barrier Function based Quadratic Programs Revisited}
\label{sec:CLF_CBF_revisited}

\subsection{Model and Input-Output Linearizing Control}
\label{subsec:IO_linearization}
Consider a nonlinear control affine hybrid model
\begin{align}
\label{fg_affine}
\mathcal{H}:& \begin{cases}
\dot x =f(x)+g(x)u, & x\notin S\\
x^+=\Delta(x^-), & x\in S
\end{cases}\\
&y=y(x), \nonumber
\end{align}
where $x \in \mathbb{R}^n$ is the system state, $u \in \mathbb{R}^m$ is the control input, $S$ is the switching surface
, and $y \in \mathbb{R}^m$ is a set of outputs.

\updateT{While our approach is can be applied for a general high relative-degree system, we focus on the most common case to mechanical systems with relative-degree 2 control output.}
If the control output $y(x)$ has relative degree 2, then the time-derivative $\dot y(x)$ will be a function of the state $x$ and not dependent on the control input $u$. Considering the second time-derivative $\ddot y$, we have:
\begin{equation}
\label{linearied_output_h}
\ddot y=\frac{\partial \dot y}{\partial x} \dot x=L_f^2 y(x) + L_g L_f y(x) u,
\end{equation}
where $L$ represents the Lie derivatives. To be more specific:
\begin{align}
\label{lie_deriv_ddy_fg}
L_f^2 y(x) \triangleq \frac{\partial \dot y}{\partial x} f(x), \qquad
L_g L_f y(x) \triangleq \frac{\partial \dot y}{\partial x} g(x).
\end{align}
If the decoupling matrix $L_g L_f y(x)$ is invertible, then 
\begin{align}
\label{expControl}
u(x, \mu) =u_{ff}(x) + (L_g L_f y(x))^{-1} \mu,
\end{align} 
%
%
with the feed-forward control input
\begin{align}
u_{ff}(x)=-(L_gL_fy(x))^{-1} L_f^2 y(x),
\end{align}
input-output linearizes the system.
The dynamics of the system \eqref{fg_affine} can then be described in terms of dynamics of the transverse variables, $\eta \in \mathbb{R}^{2m}$, and the coordinates $z \in {\cal Z}$ with ${\cal Z}$ being the co-dimension $2m$ manifold 
\begin{equation}
	\label{eq:Z}
	{\cal Z} = \{x \in \mathbb{R}^n ~|~ \eta(x) \equiv 0\}.
\end{equation}
One choice for the transverse variables is,
\begin{align}
\eta=\begin{bmatrix}
y(x)\\
\dot y(x)
\end{bmatrix}.
\end{align} 

The input-output linearized hybrid system then is,
\begin{align}
\label{eq:lineaized_sys}
\mathcal{H}^{IO}:&
\begin{cases}
\dot \eta= \bar{f}(\eta) + \bar{g}(\eta) \mu, & \\
\dot z=p(\eta,z), & (\eta,z) \notin S,\\
		\eta^+ = \Delta_{\eta}(\eta^-, z^-), & \\
		z^+=\Delta_Z(\eta^-, z^-), & (\eta^-, z^-)\in S.
\end{cases} \\
&y=y(\eta), \nonumber
\end{align}
where $z$ represents uncontrolled states \cite{RESCLF:AmGaGrSr:TAC12}, and
\begin{equation} \label{eq:f_bar_g_bar}
	\bar{f}(\eta) = F\eta, \quad \bar{g}(\eta)=G,
\end{equation}
with,
\begin{align}
	\label{FGdefns}
	F=\begin{bmatrix}
		O & I\\
		O & O
	\end{bmatrix}
	\text{ and }
	G=\begin{bmatrix}
		O\\I
	\end{bmatrix}.
\end{align}

The linear system in \eqref{eq:f_bar_g_bar} is in controllable canonical form, and a linear controller such as $\mu=-K \eta$ can be designed such that the following closed-loop system is stable
\updateT{\begin{align}
\label{eq:A_cl}
\dot \eta= (F-GK) \eta =A_{cl} \eta.
\end{align}}
A corresponding quadratic Lyapunov function \updateT{$V(\eta)=\eta^T P \eta$, with $P$ being a symmetric positive definite matrix,} can then be established through the solution of the Lyapunov equation:
\updateT{\begin{align}
	\label{eq:P_define}
	A_{cl}^T P + P A_{cl}+Q = 0,
	\end{align}
for $Q$ being any symmetric positive definite matrix.}
\subsection{Control Lyapunov Function based Quadratic Programs}
\label{subsec:CLF_QP}
\subsubsection{CLF-QP}

Instead of a linear control design $\mu=-K\eta$ in \eqref{expControl}, an alternative control design is through a control Lyapunov function $V(\eta)$, wherein a control is chosen pointwise in time such that the time deriviative of the Lyapunov function $\dot V(\eta, \mu) \le 0$, resulting in stability in the sense of Lyapunov, or $\dot V(\eta, \mu) < 0$ for asymptotic stability, or $\dot V(\eta, \mu) + \lambda V(\eta) \le 0, \lambda > 0$ for exponential stability.

 To enable directly controlling the rate of convergence, we use a \textit{rapidly exponentially stabilizing control Lyapunov function (RES-CLF)}, introduced in \cite{RESCLF:AmGaGrSr:TAC12}.  RES-CLFs provide guarantees of \emph{rapid exponential stability} for the transverse variables $\eta$. In particular, a function $\Ve(\eta)$ is a RES-CLF for the system \eqref{fg_affine} if there exist positive constants $c_1, c_2, c_3 > 0$ such that for all $0 < \epsilon < 1$ and all states $(\eta,z)$ it holds that
\begin{align}
	\label{Vineqepsilon}
	c_1 \| \eta \|^2 \leq V_{\epsilon}(\eta) \leq \dfrac{c_2}{\epsilon^2} \|\eta \| ^2,  \\
	\label{Vinfepsilon}
	\dot{V}_{\epsilon}(\eta,\mu)+ \frac{c_3}{\epsilon} V_{\epsilon}(\eta) \leq 0.
\end{align}

The RES-CLF will take the form:  
\begin{equation}
	\label{VUnscaledCoords_eps}
	\Ve(\eta) = \eta^{T}\begin{bmatrix} \frac{1}{\epsilon}I & 0 \\ 0 & I  \end{bmatrix} P \begin{bmatrix} \frac{1}{\epsilon}I & 0 \\ 0 & I  \end{bmatrix} \eta =: \eta^{T} P_\epsilon \eta,
\end{equation}
\updateT{where $P$ satisfies the Lyapunov equation \eqref{eq:P_define}}
and the time derivative of the RES-CLF
 \eqref{VUnscaledCoords_eps} is computed as
\begin{align}
	\label{eq:CLF_time_deriv}
	\dot{V}_{\epsilon}(\eta,\mu)=\frac{\partial \Ve}{\partial \eta} \dot \eta=L_{\bar{f}}V_{\epsilon}(\eta)+L_{\bar{g}}V_{\epsilon}(\eta)\mu,
\end{align}
where
\begin{align}
	\label{LfVLgVdefns2}
	L_{\bar{f}} V_{\epsilon}(\etavar)& = \frac{\partial \Ve}{\partial \eta} \bar f (\eta)= \etavar^T (F^T P_{\epsilon} + P_{\epsilon} F) \etavar,  \nonumber\\
	L_{\bar{g}} V_{\epsilon}(\etavar)& = \frac{\partial \Ve}{\partial \eta} \bar g (\eta)=2 \etavar^T P_{\epsilon} G .
\end{align}

It can be show that for any Lipschitz continuous feedback control law $\mu$ that satisfies 
\eqref{Vinfepsilon}, it holds that
\begin{eqnarray}
	\label{RES condition for eta}
	V(\eta) \le e^{-\frac{c_3}{\epsilon}t}V(\eta(0)), \quad
	\| \eta(t) \| \leq \frac{1}{\epsilon}\sqrt{\frac{c_2}{c_1}} e^{-\frac{c_3}{2 \epsilon} t} \| \eta(0) \|,
\end{eqnarray}
i.e. the rate of exponential convergence can be directly controlled with the constant $\epsilon$ through $\frac{c_3}{\epsilon}$.
%
One such controller is the CLF-based quadratic program (CLF-QP)-based controller, introduced in \cite{TorqueSaturation:GaSrAmGr:TRO14}, where $\mu$ is directly selected through an online quadratic program to satisfy \eqref{Vinfepsilon}:

\HRule
\noindent \textbf{CLF-QP}:
\begin{align} 
\label{minnorm_convexoptim}
u^*(x) =~& \underset{u,\mu}{\argmin} & & \mu^T \mu \\
& \text{s.t.} & & \dot V_{\epsilon}(\eta, \mu)+\frac{c_3}{\epsilon} V_{\epsilon}(\eta) \le 0, \mytag{CLF}\\
&&& u =u_{ff}(x) + (L_g L_f y(x))^{-1} \mu. \mytag{IO}
\end{align}
\HRule

\updateT{T}he above minimization problem is a quadratic program since the ineqality constraint on the time-derivative of the Lyapunov function can be written as a linear ineqality constraint
\begin{align}
\label{eq:clf_linear_constraint}
A_{clf}~\mu \le b_{clf},
\end{align}
where
\begin{align}
A_{clf}=L_{\bar{g}} V_{\epsilon}(\eta);\quad
 b_{clf}=-L_{\bar{f}} V_{\epsilon}(\eta)-\frac{c_3}{\epsilon}\Ve(\eta).
\end{align}
Moreover, the IO equality constraint is linear in $u, \mu$, and can be written as:
\begin{align}
A_{IO}\begin{bmatrix}
u \\ \mu 
\end{bmatrix}=b_{IO},
\end{align}
where
\begin{align}
A_{IO}=\begin{bmatrix}
I & -(L_gL_fy(x))^{-1}
\end{bmatrix}, \quad b_{IO}=u_{ff}(x).
\end{align}
This optimization is solved pointwise in time \updateT{using} 
efficient quadratic program solvers, such as CVXGEN \cite{CVXGEN:MaBo2012}, 
\updateT{at} real-time speeds over $1$ kHz.
\begin{remark}
	Note that the solution of the above QP has been shown to be Lipschitz continuous in \cite{LipschitzContinuity:MoPoAm:CDC13}.
\end{remark}

\subsubsection{CLF-QP with Constraints}
\label{sec:clf-qp-constriants}
Formulating the control problem as a quadratic program now enables us to incorporate constraints into the optimization.  These constraints could be input constraints for input saturation or state-based constraints such as friction constraints, contact force constraints, etc., for robotic locomotion and manipulation. 
These types of constraints can be expressed in a general form as
\begin{align} 
\label{eq:constraints_u}
A_c (x) u \le b_c (x).
\end{align}

The CLF-QP based controller with additional constraints then takes the form,

\vspace{-5pt}
\HRule
\noindent \textbf{CLF-QP with Constraints}:
\begin{align} \label{eq:RelaxedCLFQP}
u^*(x)=	& \underset{u,\mu, \delta}{\argmin} & & \mu^T \mu + p \delta ^2\\
& \text{s.t.} & & \dot V_{\epsilon}(\eta,\mu) + \frac{c_3}{\epsilon} V_{\epsilon}(\eta) \le \delta, \mytag{CLF}\\
&&& A_c(x) u \le b_c(x), \mytag{Constraints}\\
&&& u =u_{ff}(x) + (L_g L_f y(x))^{-1} \mu, \mytag{IO}
\end{align}
\HRule
where $p$ is a large positive number that represents the penalty of relaxing the CLF inequality, which is necessary to keep the QP feasible when the constraints conflict with each other. 

The constraints above could be input saturation constraints expressed as,
\begin{align}
\label{CLF-QP-Saturation}
u^*(x)=& \underset{u,\mu, \delta}{\argmin} & & \mu^T \mu + p \delta^2\\
& \text{s.t.} & & \dot V_{\epsilon}(\eta,\mu) + \frac{c_3}{\epsilon} V_{\epsilon}(\eta) \le \delta, \mytag{CLF}\\
&&& u_{min} \le u \le u_{max} \mytag{Input Saturation},\\
&&& u =u_{ff}(x) + (L_g L_f y(x))^{-1} \mu \mytag{IO}.
\end{align}

Note that, similar to \eqref{eq:clf_linear_constraint}, the CLF inequality condition in the above CLF-QPs is affine in $\mu$, ensuring that these are actually quadratic programs.

This formulation opened a novel method to guarantee stability of nonlinear systems with respect to additional constraints such as torque saturation \cite{TorqueSaturation:GaSrAmGr:TRO14}, wherein experimental demonstration of bipedal walking with strict input constraints was demonstrated, 
and in enabling the application of L1 adaptive control for bipedal robots \cite{NgSr2014Adaptive_ACC}.


\update{
	\begin{remark}
		\label{remark:stability-relaxed-clf-qp}
		\textbf{(Stability of CLF-QP with relaxed CLF inequality.)} In subsequent sections of this paper, we develop different types of controllers based on the relaxed CLF-QP \eqref{eq:RelaxedCLFQP}. Allowing the violation of the RES-CLF condition \eqref{Vinfepsilon} enables us to incorporate various other input and state constraints as we saw in \eqref{eq:RelaxedCLFQP}.  Additionally, the relaxed RES-CLF condition also enables incorporating safety constraints through barriers (see Section \ref{sec:CBF}), as well as enabling modification of the controller to increase the robustness of the closed-loop system (see Section \ref{sec:robust_control}). However, it must be noted that the relaxation of the RES-CLF condition could lead to potential instability.  In Section \ref{sec:proof_main_result}, we establish sufficient conditions under which the relaxed CLF-QP controller can still retain the exponential stability of the hybrid periodic orbit.
	\end{remark}
}

Having revisited control Lyapunov function based quadratic programs, we will next revisit control Barrier functions.

\subsection{Control Barrier Function}
\label{sec:CBF}

We begin with the control affine system \eqref{fg_affine}
with the goal to design a controller to keep the state $x$ in the safe set
\begin{equation} \label{eq:safe_set}
\mathcal{C}=\left\lbrace x \in \mathbb{R}^n: h(x) \ge 0\right\rbrace, 
\end{equation} 
where $h:\mathbb{R}^n \to \mathbb{R}$ is a continuously differentiable function.
Then a function $B:\mathcal{C} \to \mathbb{R}$ is a Control Barrier Function (CBF) \cite{CBF:CBFQPwithACC:CDC14} if there exists class $\mathcal{K}$ function $\alpha_1$ and $\alpha_2$ such that, for all $x\in Int(\mathcal{C})=\left\lbrace x \in \mathbb{R}^n: h(x) > 0\right\rbrace $,
\begin{align}
\frac{1}{\alpha_1(h(x))} \le B(x) \le \frac{1}{\alpha_2(h(x))}, \\
\dot B(x,u)\le \frac{\gamma}{B(x)}, \label{eq:Bdot}
\end{align}
where \updateT{$\gamma > 0$} and
\begin{align}
\label{eq:Bdot_expr}
\dot B(x,u)=\frac{\partial B}{\partial x}\dot x = L_fB(x)+L_gB(x)u,
\end{align}
with the Lie derivatives computed as,
\begin{align}
\label{lie_deriv_dB_fg}
L_fB(x)=\frac{\partial B}{\partial x}f(x),\qquad
L_gB(x)=\frac{\partial B}{\partial x}g(x).
\end{align}
Thus, if there exists a Barrier function $B(x)$ that satisfies the CBF condition in \eqref{eq:Bdot},
then $\mathcal{C}$ is forward invariant, or in other words, if $x(0)=x_0 \in \mathcal{C}$, i.e., $h(x_0) \ge 0$, then $x=x(t) \in \mathcal{C}, \forall t$, i.e., $h(x(t)) \ge 0, \forall t$.  Note that, as mentioned in \cite{CBF:CBFQPwithACC:CDC14}, this notion of a CBF is stricter than standard notions of CBFs in prior literature that only require $\dot{B} \le 0$.

In this paper, we will use the following reciprocal control Barrier candidate function:
\begin{align}
\label{CBF_candidate_inv}
B(x)=\frac{1}{h(x)}.
\end{align}

Incorporating the CBF condition \eqref{eq:Bdot} into the CLF-QP, 
%
%
%
we have the following CBF-CLF-QP based controllers:
\HRule
\noindent \textbf{CBF-CLF-QP}:
\begin{align}
\label{CBF_CLF_QP_controller}
u^*(x)=	& \underset{u,\mu,\delta}{\argmin} & & \mu^T \mu + p~\delta^2 \\
& \text{s.t.} & & \dot V_{\epsilon}(\eta, \mu)+\frac{c_3}{\epsilon} V_{\epsilon}(\eta) \le \delta,  \mytag{CLF} \\
& & &\dot B(x, u)-\frac{\gamma}{B(x)} \le 0,  \mytag{CBF} \\
&&& A_c(x) u \le b_c(x), \mytag{Constraints}\\
&&& u =u_{ff}(x) + (L_g L_f y(x))^{-1} \mu. \mytag{IO}
\end{align}
\HRule

\editQTN{\begin{remark}
	Note that, similar to \eqref{eq:clf_linear_constraint}, for any nonlinear affine system and nonlinear state-dependent constraint, the CBF condition \eqref{eq:Bdot} is scalar and affine in $u$, ensuring a compact quadratic program that can be solved in 1 kHz. More importantly, the satisfaction of this instantaneous condition at each time guarantees the state dependent constraint $h(x)\ge 0$ for all $t\ge t_0$.
\end{remark}}

Having presented control Lyapunov functions, control Barrier functions, and their incorporation into a quadratic program with constraints, we now prove the stability of the QP controller with relaxed CLF inequality. 



\section{Sufficient conditions for the stability of CLF with relaxed inequality}
\label{sec:proof_main_result}
In this Section, we will present two theorems and their proof about the stability of CLF based controller with relaxed RES-CLF condition for both continuous-time and hybrid systems.
\updateT{We recall that as discussed in Section \ref{sec:clf-qp-constriants} and Remark \ref{remark:stability-relaxed-clf-qp}, relaxing the CLF inequality constraint is necessary to keep the QP feasible when additional potentially conflicting constraints are added to the QP.  However, in this case, we could have a potential loss in stability guarantees if a certain sufficient condition, that is established next, is not satisfied.}

We begin with the standard RES-CLF that guarantees the following inequality,
\begin{equation}
\label{StandardClfInequality}
\dot{V}_{\epsilon}(\eta,\mu)+\frac{c_3}{\epsilon}V_{\epsilon}(\eta)\leq 0.
\end{equation}
The CLF-QP with the relaxed inequality takes the form,
\begin{align}
\label{RelaxedInequality}
\dot{V_{\epsilon}}(\eta,\mu)+\frac{c_3}{\epsilon}V_{\epsilon}(\eta)\leq d_{\epsilon},
\end{align}
where $d_{\epsilon}(t) \ge 0$ represents the time-varying relaxation of the RES-CLF condition.
We define
\begin{equation}
\label{W2Defn}
w_{\epsilon}(t)=\int_0^t \! \frac{d_{\epsilon}(\tau)}{V_{\epsilon}} \, \mathrm{d}\tau.
\end{equation}
to represent a scaled version of the total relaxation up to time $t$.  In the following subsections, we make use of this quantity to establish exponential stability under certain conditions for both continuous-time and hybrid systems.  In Section \ref{app:stability-continuous-systems} we will look at continuous-time systems where we will need exponential stability ($\epsilon = 1$ in the above formulations), while in Section \ref{app:stability-hybrid-systems} we will look at hybrid systems where we need rapid exponential stability ($\epsilon < 1$).

\subsection{Stability of the Relaxed CLF-QP Controller for Continuous-Time Systems}
\label{app:stability-continuous-systems}
	Consider a control affine system
	\begin{align}
		\label{eq:sys_zero_dyn}
		\dot \eta &=\bar f(\eta)+\bar g(\eta)\mu, \\ \nonumber
		\dot z &=p(\eta,z), 
	\end{align}
	where $\eta$ is the controlled (or output) state, and $z$ is the uncontrolled state. Let $\Orbitz$ be an exponentially stable periodic orbit for the zero dynamics $\dot z=p(0,z)$. Let $\Orbit=\iota_0(\Orbitz)$ be a periodic orbit for the full-order dynamics corresponding to the periodic orbit of the zero dynamics, $\Orbitz$,
	through the canonical embedding $\iota_0: Z \to X \times Z$ given by $\iota_0(z)=(0,z)$. As stated in \cite[Theorem 1]{RESCLF:AmGaGrSr:TAC12}, for all control inputs $\mu(\eta,z)$ that guarantee enforcing the ES-CLF condition (with $\epsilon = 1$) \eqref{StandardClfInequality}, 
	we have that $\Orbit = \iota_0(\Orbitz)$ is an exponentially stable periodic orbit of \eqref{eq:sys_zero_dyn}. Then, for the relaxed CLF-QP condition \eqref{RelaxedInequality} ($d_\epsilon \not\equiv 0$), the following theorem establishes sufficient conditions under which the exponential stability of the periodic orbit still holds. 
	
\begin{theorem}
	\label{thm:stability_relaxedCLF_continuous}
Consider a nonlinear control affine system \eqref{eq:sys_zero_dyn}. Let $\Orbitz$ be an exponentially stable periodic orbit for the zero dynamics $\dot z=p(0,z)$. Then $\Orbit = \iota_0(\Orbitz)$ is an exponentially stable periodic orbit of \eqref{eq:sys_zero_dyn}, if
\begin{equation}
\bar w_{\epsilon} := \sup_{t\ge 0} ~ w_{\epsilon}(t)
\end{equation}
is a finite number.
\end{theorem}

\begin{proof}
Note that if $d_{\epsilon}(t) \equiv 0$, i.e., there is no violation on the RES-CLF condition, we will have conventional exponential stability as stated in \cite{RESCLF:AmGaGrSr:TAC12}. Here, we will extend the proof of exponential stability to the case of relaxation of the control Lyapunov function condition when $\bar w_{\epsilon}$ is finite.

We begin by noting that since $d_{\epsilon}(t)\ge 0$, we have
\begin{equation}
\label{w_bound}
w_{\epsilon}(t) \le \bar w_{\epsilon}, ~ \forall t \ge 0.
\end{equation}
Next, from \eqref{RelaxedInequality}, we have
\begin{align}
	~\frac{dV_{\epsilon}}{dt} &\le -\frac{c_3}{\epsilon}V_{\epsilon}+d_{\epsilon}(t), \nonumber \\
	\Rightarrow ~\frac{dV_{\epsilon}}{V_{\epsilon}} &\le  -\frac{c_3}{\epsilon}dt+\frac{d_{\epsilon}(t)}{V_{\epsilon}}dt, \nonumber \\
	\Rightarrow ~ ln\left( \frac{V_{\epsilon}(t)}{V_{\epsilon}(0)}\right) &\le  -\frac{c_3}{\epsilon}t+\int_0^t \! \frac{d_{\epsilon}(\tau)}{V_{\epsilon}} \, \mathrm{d}\tau, \nonumber \\
%
\label{VetVe0}
\Rightarrow ~ V_{\epsilon}(t) &\le  e^{-\frac{c_3}{\epsilon}t+w_{\epsilon}(t)}V_{\epsilon}(0).
\end{align}

This, in combination with the inequality in \eqref{Vineqepsilon}, then implies that
\begin{equation}
\label{XtX0}
\|\eta(t) \| \le \sqrt{\frac{c_2}{c_1}} \frac{1}{\epsilon} e^{-\frac{c_3}{2\epsilon}t+\frac{1}{2}w_{\epsilon}(t)}  \|\eta(0) \|.
\end{equation}
Moreover, from the inequality in \eqref{w_bound}, we have
\begin{eqnarray} \label{eq:exp-stability-relaxed-CLF}
\|\eta(t) \| & \le & \sqrt{\frac{c_2}{c_1}} \frac{1}{\epsilon} e^{-\frac{c_3}{2\epsilon}t+\frac{1}{2}\bar w_{\epsilon}}  \|\eta(0) \|, \nonumber \\
& = & \left( \sqrt{\frac{c_2}{c_1}} \frac{1}{\epsilon} e^{\frac{1}{2}\bar w_{\epsilon}} \right) e^{-\frac{c_3}{2\epsilon}t}  \|\eta(0) \|.
\end{eqnarray}
Therefore, if $\bar w_{\epsilon}$ is finite, the control output $\eta\updateT{=0}$ will still be exponentially stable under the relaxed CLF condition \eqref{RelaxedInequality}. Therefore, from \cite[Theorem 1]{RESCLF:AmGaGrSr:TAC12}, $\Orbit = \iota_0(\Orbitz)$ is an exponentially stable periodic orbit of \eqref{eq:sys_zero_dyn}.
\begin{remark}
	It must be noted that the CLF $V_\epsilon$ no longer provides a guarantee of exponential stability due to the relaxation in \eqref{RelaxedInequality}.  However due to the inequality established in \eqref{eq:exp-stability-relaxed-CLF}, and by the converse Lyapunov function theorems \cite{KHA02}, there exists another CLF $\tilde V_\epsilon$ that satisfies
	$\tilde c_1 ||\eta||^2 \le \tilde V_\epsilon(\eta) \le \frac{\tilde c_2}{\epsilon^2} ||\eta||^2$, and 
	$\dot{\tilde V}_\epsilon(\eta, \mu) + \frac{\tilde c_3}{\epsilon} \tilde V_\epsilon(\eta) \le 0$, for some positive constants $\tilde c_1, \tilde c_2, \tilde c_3$, and
	guarantees exponential stability.
\end{remark}
\end{proof}


In the next section, we consider the hybrid system with impulse effects, wherein we will need to take into account the impact time or the switching time that signifies the end of the continuous-time phase and involves a discrete-time jump in the state.

\subsection{Stability of the Relaxed CLF-QP Controller for Hybrid Systems}
\label{app:stability-hybrid-systems}

Here we look at the stability of the relaxed CLF-QP controller for hybrid systems of the form as defined in \eqref{eq:lineaized_sys} without the restriction on the vector fields as in \eqref{eq:f_bar_g_bar}.
Similar to the stability analysis for continuous-time systems in the previous section, we also use here the notions of $d_{\epsilon}(t)$ \eqref{RelaxedInequality}, the relaxation of the CLF condition, and $w_{\epsilon}(t)$ \eqref{W2Defn}, the scaled version of the total relaxation up to time $t$.

We also define $T_I^\epsilon(\eta,z)$ to be the time-to-impact or time taken to go from the state $(\eta,z)$ to the switching surface $S$. Then, intuitively, $w_\epsilon(T_I^\epsilon(\eta,z))$ indicates a scaled version of the total violation of the RES-CLF bound in \eqref{StandardClfInequality} over one complete step. If $w_\epsilon(T_I^\epsilon(\eta,z)) \le 0$, it implies that $V_\epsilon(T_I^\epsilon(\eta,z))$ is less than or equal to what would have resulted if the RES-CLF bound had not been violated at all.  As we will see in the following theorem, we will in fact only require $w_\epsilon(T_I^\epsilon(\eta,z))$ to be upper bounded by a positive constant for exponential stability.

We first define the hybrid zero dynamics as the hybrid dynamics \eqref{eq:lineaized_sys} restricted to the surface $Z$ in \eqref{eq:Z}, i.e.,
\begin{align}
	\label{eq:hybrid_zero_dyn}
	\mathcal{H}_Z:&
	\begin{cases}
		\dot z = p(0,z), & z \notin S \cap Z, \\
		z^+ = \Delta_z(0,z^-), & z^- \in S \cap Z.
	\end{cases}
\end{align}
We then have the following theorem:
\begin{theorem}
	\label{thm:MainRelaxedIneq}
	{\it Let $\Orbitz$ be an exponentially stable hybrid periodic orbit of the hybrid zero dynamics $\HS|_Z$ \eqref{eq:hybrid_zero_dyn} transverse to $\S \cap Z$ and the continuous dynamics of ${\mathscr H}$ \eqref{eq:lineaized_sys} controlled by a CLF-QP with relaxed inequality \eqref{eq:RelaxedCLFQP}. Then there exists an $\overline{\epsilon} > 0$ and an $\bar{w}_{\epsilon} \ge 0$ such that for each $0 < \epsilon < \overline{\epsilon}$, if the solution $\mu_\epsilon(\eta,z)$ of the CLF-QP \eqref{eq:RelaxedCLFQP} satisfies $w_{\epsilon}(\TI) \le \bar{w}_{\epsilon}$, then $\Orbit = \iota_0(\Orbitz)$ is an exponentially stable hybrid periodic orbit of ${\mathscr H}$.}
\end{theorem}
\gap
\begin{proof}
	See Appendix \ref{sec:proof_2_appendix}. 
\end{proof}

Having presented the stability analysis for the QP controller with relaxed CLF inequality, in the next Section, we will explore the effects of model uncertainty and propose robust QP controllers to guarantee stability, safety and constraints under model uncertainty.
	
\section{Proposed Robust Control based Quadratic Programs}
\label{sec:robust_control}

The controllers presented in Section \ref{sec:CLF_CBF_revisited} provide means of an optimal control scheme that balances conflicting requirements between stability, state-based constraints and energy consumption. It is a powerful method that has been deployed successfully for different applications, for example Adaptive Cruise Control \cite{CBF:Ames:ExperimentalACC:ACC15}, quadrotor flight \cite{CBF:Guofan:SafetyCriteria:ACC15}, and dynamic walking for bipedal robots \cite{CBF:Quan:Footstep:ADHS15}, \cite{CBF:Ames:CBFBiped:ACC15}.

However, a  primary disadvantage of the CBF-CLF-QP controller is that it requires the knowledge of an accurate dynamical model of the plant to make guarantees on stability as well as safety. In particular, enforcing the CLF and CBF conditions in the QP \eqref{CBF_CLF_QP_controller} requires the accurate knowledge of the nonlinear dynamics $f(x), g(x)$ to compute the Lie derivatives in \eqref{lie_deriv_ddy_fg}, \eqref{LfVLgVdefns2} and \eqref{lie_deriv_dB_fg}. Therefore, model uncertainty that is usually present in physical systems can potentially cause poor quality of control leading to tracking errors, and lead to instability \cite{NgSr2014Robust_RSS}, as well as violation of the safety-critical constraints \cite{CBF:Ames:RobustnessCBF:ADHS15}.  In this section, we will explore the effect of uncertainty on the CLF-QP controller, input and state constraints, and safety constraints enforced by the CBF-QP controller and then proposed robust controllers based quadratic programs to address those effects.
Firstly, we will analyze the effects of model uncertainty on the CLF condition.

\subsection{Adverse Effects of Uncertainty on the CLF-QP controller}
\label{sec:adverse_effect}

\begin{table}
	\centering
	\begin{tabular}{ |c|c|c| } 
		\hline
		Notations & Model types\\
		\hline
		{\large $f,g$} & true \textit{(unknown)} nonlinear model\\ 
		{\large $\tilde{f}, \tilde{g}$} & nominal nonlinear model\\ 
		{\large $\bar f, \bar g$} & true \textit{(unknown)} I-O linearized model\\
		{\large $\tilde{\bar f}, \tilde{\bar g}$} & nominal I-O linearized model \\
		\hline
	\end{tabular}
	\caption{A list of notations for different models used in this paper. A true model represents the actual (possibly not perfectly known) model of the physical system, while the nominal model represents the model that the controller uses.  While most controllers assume the true model is known, the robust controllers in this paper use the nominal models and offer robustness guarantees to the uncertainty between the two models.}
	\label{tab:ModelTypes}
\end{table}

In order to analyze the effects of model uncertainty in our controllers, we assume that the vector fields, $f(x), g(x)$ of the real dynamics \eqref{fg_affine}, are unknown. We therefore have to design our controller based on the nominal vector fields $\tilde{f}(x), \tilde{g}(x)$.  Then, the pre-control law \eqref{expControl} get's reformulated as 
\begin{equation}
\label{preControlUncertainty}
u(x) = \tilde u_{ff}(x) + (L_{\tilde{g}}L_{\tilde{f}} y(x))^{-1}\mu,
\end{equation}
with
\begin{equation}
\label{ustarUncertainty}
\tilde u_{ff}(x) := -(L_{\tilde{g}}L_{\tilde{f}} y(x))^{-1} L^2_{\tilde{f}} y(x),
\end{equation}
where we have used the nominal model rather than the unknown real dynamics.

Substituting $ u(x) $ from \eqref{preControlUncertainty} into \eqref{linearied_output_h}, the input-output linearized system then becomes
\begin{equation}
\label{yddotUncertainty}
\ddot y = \mu+\Delta_{1}+\Delta_{2}\mu,
\end{equation}
where $\Delta_1 \in \mathbb{R}^m, \Delta_2 \in \mathbb{R}^{m\times m}$, are given by,
\begin{align}
\label{DeltaUncertainty}
\Delta_{1}&:=L^2_{f}\updateT{y(x)}-L_{g}L_{f} \updateT{y(x)}(L_{\tilde{g}}L_{\tilde{f}} \updateT{y(x)})^{-1}L^2_{\tilde{f}}\updateT{y(x)},  \nonumber\\
\Delta_{2}&:=L_{g}L_{f} \updateT{y(x)}(L_{\tilde{g}}L_{\tilde{f}} \updateT{y(x)})^{-1}-I.
\end{align}
\begin{remark}
	In the definitions of $\Delta_1, \Delta_2$, note that when there is no model uncertainty, i.e., $\tilde{f}=f, \tilde{g}=g$, then $\Delta_1 = \Delta_2 = 0$.
\end{remark}

Using $F$ and $G$ as in \eqref{FGdefns}, the closed-loop system now takes the form
\begin{equation}
\label{EtaClosedLoopUncertainty}
\dot{\eta}=\tilde{\bar f}(\eta)+\tilde{\bar g}(\eta)\mu,
\end{equation}
where
\begin{align}
\tilde{\bar f}(\eta)=F\eta+\begin{bmatrix} O \\ \Delta_1 
\end{bmatrix},
\tilde{\bar g}(\eta)=G+\begin{bmatrix} O \\ \Delta_2 
\end{bmatrix}.
\end{align}
In fact for $\Delta_1 \ne 0$, the closed-loop system does not have an equilibrium, and for $\Delta_2 \ne 0$, the controller could potentially destabilize the system (see \cite{reed1989instability}).
This raises the question of whether it's possible for controllers to account for this model uncertainty, and if so, how do we design such a controller. In particular, the time-derivative of the CLF in \eqref{eq:CLF_time_deriv} becomes more complex and dependent on $\Delta_1, \Delta_2$.

\editQTN{
Again, because our controller can only use the nominal model $\tilde{\bar f}, \tilde{\bar g}$ in \eqref{EtaClosedLoopUncertainty} and not the true model $\bar f, \bar g$ as defined in \eqref{eq:f_bar_g_bar} (see Table \ref{tab:ModelTypes} for different types of models),  the Lie derivatives of $V_{\epsilon}$ \editQTN{with respect to the nominal model $\tilde{\bar f}, \tilde{\bar g}$} will be as follows:
\begin{align}
L_{\tilde {\bar f}}V_{\epsilon}&=\etavar^T (F^T P_{\epsilon} + P_{\epsilon} F) \etavar, \\
L_{\tilde {\bar g}}V_{\epsilon}&=2 \etavar^T P_{\epsilon} G.
\end{align}
}
\editQTN{
Then, the true time-derivative of the CLF defined in \eqref{VUnscaledCoords_eps} is:
\begin{align}
\label{eq:dV_uncertainty}
\dot V_{\epsilon} &= L_{\bar f} V_{\epsilon} + L_{\bar g}V_{\epsilon} \mu \nonumber \\
& = L_{\tilde{\bar f}}V_{\epsilon} 
+ \underbrace{2\eta ^T P_{\epsilon} \begin{bmatrix}\textbf{0}\\\Delta_1\end{bmatrix} }_{\Delta_1^v} 
+ \underbrace{L_{\tilde{\bar g}}V_{\epsilon} \mu}_{\mu_v} 
+ \underbrace{L_{\tilde {\bar g}}V_{\epsilon}\Delta_2 \mu }_{\Delta_2^v \mu_v}\nonumber \\
& =L_{\tilde {\bar f}}V_{\epsilon} + \Delta_1^v + (1 + \Delta_2^v) \mu_v,
\end{align}
where we have defined the following new scalar variables: uncertainty $\Delta_1^v \in \mathbb{R}$, virtual control input $\mu_v \in \mathbb{R}$, and uncertainty $\Delta_2^v\in \mathbb{R}$.
}

\subsection{Robust CLF-QP}
Having discussed the effect of model uncertainty on the control Lyapunov function based controllers, we now develop a robust controller that can guarantee
tracking and stability in the presence of bounded uncertainty.
As we will see, both stability and tracking performance (rate of convergence and errors going to zero) are still retained for all uncertainty within a particular bound.   
For uncertainty that exceeds the specified bound, there is graceful degradation in performance.

\editQTN{
With the presence of model uncertainty, the CLF condition \eqref{Vinfepsilon} now becomes:
\begin{align}
\label{eq:CLF_uncertainty}
\dot V_{\epsilon} (\eta, \Delta_1^v, \Delta_2^v, \mu_v) + \frac{c_3}{\epsilon}V_{\epsilon} \le 0.
\end{align}
In general, we can not satisfy this inequality 
for all possible unknown $\Delta_1^v, \Delta_2^v$.  To address this, we assume the uncertainty is bounded as follows
\begin{align}
\label{assump:BoundUncertainty}
|\Delta_1^v| \le \Delta_{1,max}^v, \qquad |\Delta_2^v| \le \Delta_{2,max}^v.
\end{align}
}
\editQTN{
Under this assumption, we have the following robust CLF condition:
\begin{align}
\label{eq:robust_RES_cond}
&\dot V_{\epsilon} (\eta, \Delta_1^v, \Delta_2^v, \mu_v) + \frac{c_3}{\epsilon}V_{\epsilon} \le 0, \nonumber \\
& \qquad \forall |\Delta_1^v|\le\Delta_{1,max}^v, \quad \forall  |\Delta_2^v|\le\Delta_{2,max}^v \nonumber\\
\Leftrightarrow &\max_{\substack{|\Delta_1^v|\le \Delta_{1,max}^v \\ |\Delta_2|\le \Delta_{2,max}^v }} \dot V_{\epsilon} (\eta, \Delta_1^v, \Delta_2^v, \mu_v) + \frac{c_3}{\epsilon}V_{\epsilon} \le 0.
\end{align}
It means that choosing $\mu$ that satisfies \eqref{eq:robust_RES_cond} implies \eqref{RES condition for eta} for every $\Delta_1^v,\Delta_2^v$ satisfying $|\Delta_1^v|\le\Delta_{1,max}^v, |\Delta_2^v|\le\Delta_{2,max}^v$.
}
\editQTN{
Then, our optimal robust control can be expressed as the following min-max problem:
\HRule
\noindent \textbf{Robust CLF}:
\begin{align}
\label{eq:robust_CLF_min_max}
u^*(x)= & \underset{u,\mu, \mu_v}{\argmin} & & \mu^T \mu\\
& \text{s.t.} & & \max_{\substack{|\Delta_1^v|\le \Delta_{1,max}^v \\ |\Delta_2^v|\le \Delta_{2,max}^v }} \dot V_{\epsilon} (\eta, \Delta_1^v, \Delta_2^v, \mu_v) + \frac{c_3}{\epsilon}V_{\epsilon} \le 0 \nonumber \\
&&& \mu_v= L_{\tilde{\bar{g}}}V_{\epsilon}\mu \nonumber \mytag{Robust CLF}\\
&&& u =\tilde u_{ff}(x) + (L_{\tilde{g}} L_{\tilde f} y(x))^{-1} \mu. \mytag{IO}
\end{align}
\HRule
}
\editQTN{
Note that, one advantage of CLF based control is to convert the problem of driving a vector output $\eta \to 0$ to the problem of driving a scalar $V_{\epsilon}(\eta) \to 0$. Therefore, since $V_{\epsilon}$ is a scalar, we always have $\mu_v, \Delta_{1}^v, \Delta_{2}^v$ in \eqref{eq:dV_uncertainty} being scalars. Then the min-max optimization problem \eqref{eq:robust_CLF_min_max} can be expressed as a quadratic program as follows:
\begin{align}
 &\max_{\substack{|\Delta_1^v|\le \Delta_{1,max}^v \\ |\Delta_2|\le \Delta_{2,max}^v }} \dot V_{\epsilon} (\eta, \Delta_1^v, \Delta_2^v, \mu_v) + \frac{c_3}{\epsilon}V_{\epsilon} \le 0 \nonumber 
 \\ \Leftrightarrow & \max_{\substack{|\Delta_1^v|\le \Delta_{1,max}^v \\ |\Delta_2|\le \Delta_{2,max}^v }} L_{\tilde {\bar f}}V_{\epsilon} + \Delta_1^v + (1 + \Delta_2^v) \mu_v + \frac{c_3}{\epsilon}V_{\epsilon} \le 0,
\end{align}
and the above robust CLF condition holds if
\begin{align}
\begin{cases}
{\Psi}_{0}^{ max}+{\Psi}_{1}^p \mu_v \le 0
\\{\Psi}_{0}^{ max}+{\Psi}_{1}^n \mu_v \le 0
\end{cases},
\end{align}
where
\begin{align}
{\Psi}_{0}^{max}&= L_{\tilde{\bar f}}V\updateT{_{\epsilon}}+\frac{c_3}{\epsilon}V_{\epsilon}+ \Delta_{1,max}^v \nonumber\\
{\Psi}_{1}^p& = 1+\Delta_{2,max}^v,\nonumber\\
{\Psi}_{1}^n& =1-\Delta_{2,max}^v.
\end{align}
Thus the robust optimal CLF in \eqref{eq:robust_CLF_min_max} can be satisfied using the following quadratic program based controller:
%
%
\HRule
\noindent \textbf{Robust CLF-QP}:
\begin{align}
\label{cvxWithUncertainty}
u^* =& \underset{u,\mu,\mu_v}{\argmin} & & \mu^T \mu \\
& \text{s.t.} & & {\Psi}_{0}^{max}(\eta,\Delta_{1,max}^v) + {\Psi}_{1}^p(\Delta_{2,max}^v)~\mu_v \le 0, \nonumber\\ & & & {\Psi}_{0}^{max}(\eta,\Delta_{1,max}^v) + {\Psi}_{1}^n(\Delta_{2,max}^v)~\mu_v \le 0, \nonumber\\
&&& \mu_v= L_{\tilde{\bar{g}}}V_{\epsilon}\mu, \nonumber \mytag{Robust CLF}\\
&&& u =\tilde u_{ff}(x) + (L_{\tilde{g}} L_{\tilde f} y(x))^{-1} \mu. \mytag{IO}
\end{align}
\HRule
}
\begin{remark} \label{remark:local-feasibility}
	Note that since $V(\eta)$ is a local CLF for the nonlinear system \eqref{fg_affine} and the fact that $\dot{V}$ depends continuously on $\Delta_1^v, \Delta_2^v$, there always exists $\mu$ satisfying \eqref{eq:robust_RES_cond} for a local region around $\eta = 0, \Delta_1^v = 0, \Delta_2^v = 0$.  Moreover, if $\Delta_{1,max}^v, \Delta_{2,max}^v$ are chosen to be within this region, then the above QP is guaranteed to be feasible.
	\updateT{For QPs with additional constraints we assume point-wise feasibility.}
\end{remark}

\begin{remark}
	Note that $\Delta_{1,max}^v, \Delta_{2,max}^v$ are design parameters for control design\updateT{, and can be easily changed online}.  While we do not present a principled way of choosing these parameters, it must be noted that $\Delta_{1,max}^v, \Delta_{2,max}^v$ can be chosen based on knowledge of the expected uncertainty the controller is to be designed to handle.
\end{remark}

Having developed the robust version of the CLF-QP based controller, we now incorporate constraints into the robust control formulation.  We note that this first formulation is for non-robust constraints, i.e., these constraints are only evaluated on the nominal model available to the controller.  Thus this control is valid and makes sense for only those constraints that are not dependent on the true model.

The incorporation of constraints into the CLF-QP controller required relaxation of the CLF condition (see \cite{TorqueSaturation:GaSrAmGr:Access15}).  Similarly, here we relax the robust CLF condition to obtain,
\HRule
\noindent \textbf{Robust CLF-QP with Constraints}:
\begin{align}
\label{Robust_CLF_QP_constraint}
u^*(x)= &\underset{u,\mu, \mu_v, \delta_1, \delta_2}{\argmin}   \mu^T \mu + p_1 \delta_1 ^2 + p_2 \delta_2^2\\
 \text{s.t.}~~~ &  {\Psi}_{0}^{max}(\eta,\Delta_{1,max}^v) + {\Psi}_{1}^p(\Delta_{2,max}^v)~\mu_v \le \delta_1, \nonumber\\ 
 & {\Psi}_{0}^{max}(\eta,\Delta_{1,max}^v) + {\Psi}_{1}^n(\Delta_{2,max}^v)~\mu_v \le \delta_2, \nonumber\\
& \mu_v= L_{\tilde{\bar{g}}}V_{\epsilon}\mu, \nonumber \mytag{Robust CLF}\\
& A_c(x)u \le b_c(x), \nonumber \mytag{Constraints}\\
& u =u_{ff}(x) + (L_g L_f y(x))^{-1} \mu. \mytag{IO}
\end{align}
\HRule

\begin{remark}
	It's critical to note that the additional constraints in \eqref{Robust_CLF_QP_constraint} are not robust, and that this method is only applicable for constraints that are invariant to model uncertainty.  Constraints such as torque saturation, like in \eqref{CLF-QP-Saturation}, are invariant to model uncertainty since the control inputs are computed directly from the nominal model $\tilde f(x), \tilde g(x)$, and do not depend on the real model $f(x), g(x)$. However, constraints that are invariant to the model uncertainty are not common, and we will have to explicitly address and capture the effect of uncertainty on constraints. We will address this in Section \ref{subsec:robust_constraints}.
\end{remark}

\subsection{Robust CBF}
\label{subsec:robust_CBF}
Following the same approach we use for the robust CLF, we now analyze the effect of model uncertainty on the CBF condition and then present our robust CBF controller.

Similar to what we have seen about the effect of uncertainty on CLFs, we will now see the effect of uncertainty on CBFs. We note that the time-derivative of the Barrier function in \eqref{CBF_candidate_inv} depends on the real model. Therefore we need to enforce the following constraint given by \eqref{eq:Bdot}:
\begin{align}
\label{CBF_cond_real}
\dot B(x,f,g,u)=L_fB(x)+L_gB(x)u\le \frac{\gamma}{B(x)}
\end{align}
where $\dot x=f(x)+g(x)u$ is the real system dynamic.
As seen in the case of control Lyapunov functions and constraints, naively enforcing this barrier constraint using the nominal model results in,
\begin{align}
\label{CBF_cond_nominal_model}
\dot B(x,\tilde f,\tilde g,u)=L_{\tilde{f}}B(x)+L_{\tilde g}B(x)u\le \frac{\gamma}{B(x)}
\end{align}
where $\dot x=\tilde f(x)+\tilde g(x)u$ is the nominal system dynamics known by the controller.
Clearly due to model uncertainty, or the difference between $(f(x),g(x))$ and $(\tilde f(x), \tilde g(x))$, the constraint in \eqref{CBF_cond_nominal_model} is different from the one in \eqref{CBF_cond_real}.  In fact, as analyzed in \cite{CBF:Ames:RobustnessCBF:ADHS15}, this results in violation of the safety-critical constraint established by the Barrier function.

We then define a virtual control input $\mu_b \in \mathbb{R}$ such that the time derivative of the CBF is
\begin{align}
\label{eq:virtual_CBF_input}
\dot B(x, \updateT{\tilde f, \tilde{g}}, u) = L_{\updateT{\tilde{f}}} B(x)+L_{\updateT{\tilde{g}}}B(x)u =: \mu_b.
\end{align}

We call this virtual input-output linearization (VIOL). The CBF now takes the form of a linear system, $\dot B = \mu_b$, and therefore the effect of uncertainty can be easily addressed by using the same approach as with the robust CLF-QP. In particular, similar to what we have done with CLF in Section \ref{sec:adverse_effect}, the effect of model uncertainty on CBF transforms \eqref{CBF_cond_real} into the following form:
\begin{align}
\label{eq:linearized_CBF_Uncertainty}
\dot B(x,\Delta_1^b,\Delta_2^b,\mu_b) &= L_fB(x) + L_gB(x) u \nonumber \\ &=\mu_b+\Delta_1^b+\Delta_2^b\mu_b,
\end{align}
with $\Delta_1^b, \Delta_2^b \in \mathbb{R}$ and $\mu_b$ is defined based on the nominal model $\tilde{f}, \tilde{g}$:
\begin{align}
\mu_b=L_{\tilde{f}}B(x)+L_{\tilde g}B(x)u.
\end{align}


%
The CBF condition \eqref{CBF_cond_real} under model uncertainty becomes
\begin{align}
\dot B (x, \Delta_1^b, \Delta_2^b, \mu_b)-\frac{\gamma}{B(x)}\le 0.
\end{align}
Because we developed the CLF and CBF having the similar form of I-O linearized system, we now have a systematic way to develop the Robust CBF-CLF-QP. Again, we will assume that our model uncertainty is bounded, i.e.,
\begin{align} \label{eq:CBF_uncertainty_bound}
|\Delta_1^b| \le \Delta^b_{1,max}, \qquad
|\Delta_2^b| \le \Delta^b_{2,max}.
\end{align}
\editQTN{
Then, we will have the robust CBF condition based on this assumption as,
\begin{align}
\label{eq:robust_CBF_cond}
\max_{\substack{|\Delta_1^b|\le \Delta^b_{1,max} \\ |\Delta_2^b|\le \Delta^b_{2,max} }} \dot B (x, \Delta_1^b, \Delta_2^b, \mu_b)-\frac{\gamma}{B(x)}\le 0.
\end{align}
}
\editQTN{
Note that choosing $\mu_b$ that satisfies \eqref{eq:robust_CBF_cond} implies $x \in \mathcal{C}$ or $B(x)\ge 0$ for every $\Delta_1^b,\Delta_2^b$ satisfying $|\Delta_1^b|\le\Delta_{1,max}^b, |\Delta_2^b|\le\Delta_{2,max}^b$. Here, $\mathcal{C}$ is as defined in \eqref{eq:safe_set} and $B(x)$ is as defined in \eqref{CBF_candidate_inv}.
}
\editQTN{
Note that the robust CBF condition in \eqref{eq:robust_CBF_cond} is affine in $\mu_b$ and can be expressed as
	\begin{equation} 
	\label{eq:RobustCBF-Psi}
	\max_{\substack{|\Delta_1^b|\le \Delta^b_{1,max} \\ |\Delta_2^b|\le \Delta^b_{2,max} }} 
	\Psi_0^b+ \Psi_1^b \mu_b \le 0
	\end{equation}
	where ``$b$'' refers to the CBF constraint, and
	\begin{align}
	\Psi_0^b(x, \Delta_1^b)&:=\Delta_1^b-\frac{\gamma}{B(x)}, \nonumber\\
	\Psi_1^b(x, \Delta_2^b)&:=1+\Delta_2^b,
	\end{align}
	where the above arises due to the time-derivative of the CBF from \eqref{eq:linearized_CBF_Uncertainty}.
}
\editQTN{	
	Since $\Psi_0^b$, $\Psi_1^b$ are affine with respect to $\Delta_1^b$, $\Delta_2^b$, and with the assumptions on the bounds on the uncertainty in \eqref{eq:CBF_uncertainty_bound}, the robust CBF condition \eqref{eq:RobustCBF-Psi} will hold if the following two inequalities hold
	\begin{align}
	\Psi_{0,max}^b(x)+ \Psi_{1,p}^b(x) \mu_b \le 0, \nonumber\\
	\Psi_{0,max}^b(x)+ \Psi_{1,n}^b(x) \mu_b \le 0.
	\end{align}
	where
	\begin{align}
	\Psi_{0,max}^b&:=\Delta_{1,max}^b-\frac{\gamma}{B(x)},\\
	\Psi_{1,p}^b&:=1 + \Delta_{2,max}^b,\\
	\Psi_{1,n}^b&:=1 - \Delta_{2,max}^b,
	\end{align}
}	
	We then can incorporate the above robust CBF conditions into the robust CLF-QP \eqref{Robust_CLF_QP_constraint} resulting in a quadratic program.
\updateT{\begin{remark}
While we use reciprocal control barrier function \eqref{CBF_candidate_inv} in this paper, our proposed approach on robust CBF does not depend on the choice of barrier function. To be more specific, the virtual control input $\mu_b$ \eqref{eq:virtual_CBF_input} and the effect of model uncertainty on CBF \eqref{eq:linearized_CBF_Uncertainty} are defined for a general nonlinear function $B(x)$.
\end{remark}}	

Having robustified CLFs and CBFs, we will now apply this to obtain robust constraints.

\subsection{Robust Constraints}
\label{subsec:robust_constraints}
As developed in Section \ref{sec:CLF_CBF_revisited}, the input and state constraints can be expressed as \eqref{eq:constraints_u}.  These constraints depend on the model explicitly and are constraints on the real system dynamics.  We can thus rewrite the constraints showing explicit model dependence as,
\begin{align} \label{eq:constraint_true_model}
A_c(x,f,g) u \le b_c(x,f,g).
\end{align}
If a controller naively enforces these constraints using the nominal model available to the controller, the controller will enforce the constraint
\begin{align}
\label{input_constraint_nominal_model}
A_c(x,\tilde{f},\tilde{g}) u \le b_c(x,\tilde{f},\tilde{g}).
\end{align}
Clearly, the constraint in \eqref{input_constraint_nominal_model} is different from the desired constraint we want to enforce on the real model \eqref{eq:constraint_true_model}.  For instance, to enforce a contact force constraint, if the controller computes and enforces the contact force constraint using the nominal model, there is absolutely no guarantee that the actual contact force on the real system satisfies the constraint.

\begin{remark}
	As noted earlier, certain constraints do not depend on the model at all.  In such cases, model uncertainty does not affect the constraint.  One example of such a constraint is a pure input constraint, such as $u(x) \le u_{max}$.  Expressing this constraint in the form of \eqref{eq:constraint_true_model} results in $A_c = I, b_c = u_{max}$, which clearly is not dependent on the model.
\end{remark}

To robustify the ``constraints'', we can use a similar method as we did for the control barrier functions.  We start by  reformulating the constraints \eqref{eq:constraint_true_model} by using Virtual Input-Output Linearization (VIOL) for to obtain,
\begin{align}
\label{mu_c_VIOL}
	A_{c,i}(x,f,g)u-b_{c,i}(x,f,g)=:\mu_{c,i},
\end{align}
and then enforce
\begin{equation}
\label{mu_c_constraint}
	\mu_{c,i} \le 0,
\end{equation}
where the index $i$ indicates the $i^{th}$ constraints in \eqref{eq:constraint_true_model}.

With model uncertainty, we now have,
\begin{align} \label{eq:constraints-mu_c}
A_{c,i}(x,f,g)u-b_{c,i}(x,f,g)=\mu_{c,i} + \Delta_1^{c,i} + \Delta_2^{c,i} \mu_{c,i} \nonumber \\
=:\bar{\mu}_{c,i}(\mu_{c,i}, \Delta_1^{c,i}, \Delta_2^{c,i}),
\end{align}
with $\Delta_1^{c,i}, \Delta_2^{c,i} \in \mathbb{R}$ and the virtual input $\mu_{c,i}$ is now defined based on the nominal model $\tilde f, \tilde g$:
\begin{align}
\mu_{c,i}=A_{c,i}(x,\tilde f,\tilde g)u-b_{c,i}(x,\tilde f, \tilde g).
\end{align}
The robust constraint now becomes
\begin{equation}
	\bar{\mu}_{c,i}(\mu_{c,i}, \Delta_1^{c,i}, \Delta_2^{c,i}) \le 0.
\end{equation}

Once again, we make the assumption on bounded uncertainty,
\begin{align}
|\Delta_1^{c,i}| \le \Delta^{c,i}_{1,max}, \qquad
|\Delta_2^{c,i}| \le \Delta^{c,i}_{2,max},
\end{align}
such that the robust constraint condition becomes,
\begin{align}
\label{eq:robust_constraint_cond}
\max_{\substack{|\Delta_1^{c,i}|\le \Delta^{c,i}_{1,max} \\ |\Delta_2^{c,i}|\le \Delta^{c,i}_{2,max} }} \bar{\mu}_{c,i}(\mu_{c,i}, \Delta_1^{c,i}, \Delta_2^{c,i}) \le 0.
\end{align}

Similar to robust CLF and robust CBF, the above robust constraints condition is affine in $\mu_{c,i}$ and can be expressed as
\begin{equation} \label{eq:Robust-constraints-Psi}
\max_{\substack{|\Delta_1^{c,i}|\le \Delta^{c,i}_{1,max} \\ |\Delta_2^{c,i}|\le \Delta^{c,i}_{2,max} }} 
\Psi_0^{c,i}+ \Psi_1^{c,i} \mu_{c,i} \le 0
\end{equation}
where ``$(c,i)$'' refers to the $i^{th}$ constraint, and
\begin{align}
\Psi_0^{c,i}(x, \Delta_1^{c,i})&:=\Delta_1^{c,i}, \nonumber\\
\Psi_1^{c,i}(x, \Delta_2^{c,i})&:=1+\Delta_2^{c,i},
\end{align}
where the above arises due to the form of the constraints in \eqref{eq:constraints-mu_c}.  The same procedure 
 \updateT{used} for the robust CBF in the previous section can be 
 \updateT{applied} for the robust constraints as well to show how the max inequality gets converted to a set of linear inequalities 
 \updateT{that are enforced by a quadratic program.}

\subsection{Robust CBF-CLF-QP with Robust Constraints}
	We finally can unify the robust CLF for stability, robust CBF for safety enforcement, and the robust constraints  under model uncertainty to obtain the following unified robust controller.
\HRule
\noindent \textbf{Robust CBF-CLF-QP with Robust Constraints}:
\begin{align} \label{robust_CBF_CLF_QP_robust_constraint}
& u^*(x)=\\	
& \underset{u, \mu, \mu_v, \mu_b, \mu_c, \delta}{\argmin} & & \mu^T \mu + p\delta^2 \nonumber \\
	& \text{s.t.} & & \max_{\substack{|\Delta_1^v|\le \Delta_{1,max}^v \\ |\Delta_2^v|\le \Delta_{2,max}^v }} \dot V_{\epsilon} (\eta, \Delta_1^v, \Delta_2^v, \mu_v) + \frac{c_3}{\epsilon}V_{\epsilon} \le \delta, \nonumber \\
	&&& L_{\tilde{\bar g}} V_{\epsilon}\mu \mytag{Robust CLF}=\mu_v,\\
	&&&\max_{\substack{|\Delta_1^b|\le \Delta^b_{1,max} \\ |\Delta_2^b|\le \Delta^b_{2,max} }} \dot B (x, \Delta_1^b, \Delta_2^b, \mu_b)-\frac{\gamma}{B(x)}\le 0,\nonumber\\
	&&& L_{\tilde{f}}B(x)+L_{\tilde{g}}B(x)u=\mu_b, \mytag{Robust CBF} \nonumber\\
	&&&\max_{\substack{|\Delta_1^c|\le \Delta^c_{1,max} \\ |\Delta_2^c|\le \Delta^c_{2,max} }} \bar{\mu}_c(\mu_c, \Delta_1^c, \Delta_2^c) \le 0,\nonumber\\
	&&& A_c(x)\mu-b_c(x)=\mu_c, \mytag{Robust Constraints} \nonumber\\
		&&& u =\tilde u_{ff}(x) + (L_{\tilde{g}} L_{\tilde{f}} y(x))^{-1} \mu. \mytag{IO}
\end{align}
\HRule


Having presented our proposed optimal robust controller that can address stability and strictly enforce constraints under model uncertainty, we now validate our controller in simulations on \updateT{a} bipedal robot and experiments on \updateT{a} spring cart. 

\section{Simulation and Experimental Validation}
\label{sec:simulation}

\subsection{RABBIT Bipedal Robot}

To demonstrate the effectiveness of the proposed robust CBF-CLF-QP controller, we will conduct numerical simulations on the model of RABBIT, a planar five-link bipedal robot. Further description of RABBIT and the associated mathematical model can be found in \cite{CHABAOPLWECAGR02,WEBUGR04}. 


We consider model uncertainty in bipedal robotic walking by adding an unknown heavy load to the torso of the RABBIT robot to validate the performance of our proposed robust controllers. We will also require enforcement of contact force constraints (state constraints) and foot-step location constraints (safety constraints) in the presence of the model uncertainty.

\editQTN{In the following simulations, we run an offline optimization process to generate a walking gait for the nominal system.  This results in a set of outputs (virtual constraints) that need to be regulated to zero by the controller.  Although the offline optimization is on the nominal system, as we will see, the robust controller is able to guarantee enforcement of the virtual constraints on the true model while subject to input constraints, contact force constraints, safety constraints, and model uncertainty.}

\begin{figure*}
	\centering
	\subfloat
	{\resizebox{0.06\linewidth}{!}{\includegraphics{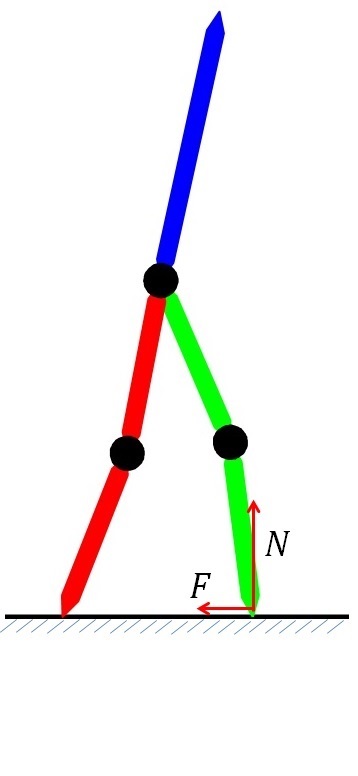}}}
	\subfloat
	{\resizebox{0.22\linewidth}{!}{\includegraphics{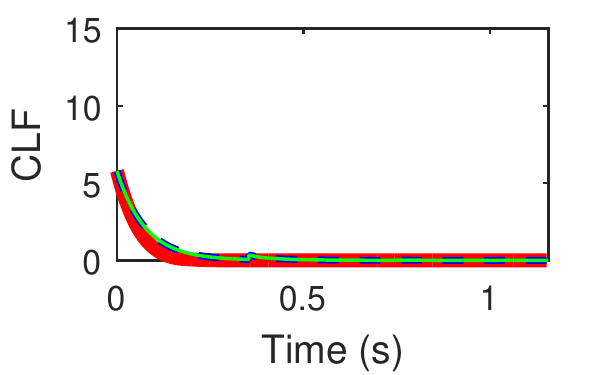}}}
	\subfloat
	{\resizebox{0.22\linewidth}{!}{\includegraphics{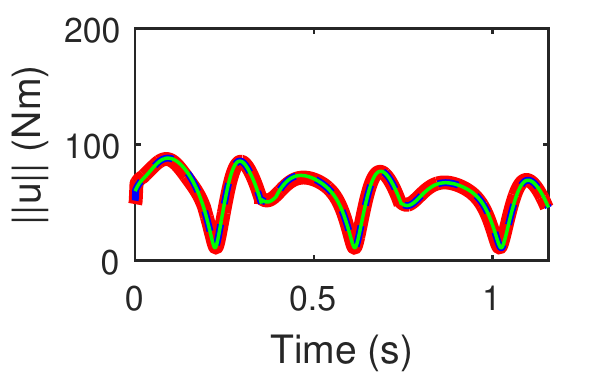}}}
	\subfloat
	{\resizebox{0.22\linewidth}{!}{\includegraphics{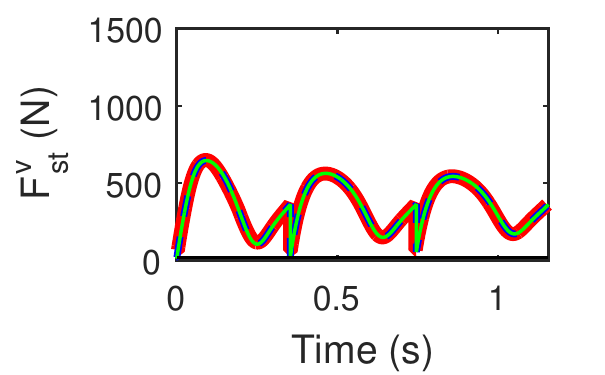}}}
	\subfloat
	{\resizebox{0.22\linewidth}{!}{\includegraphics{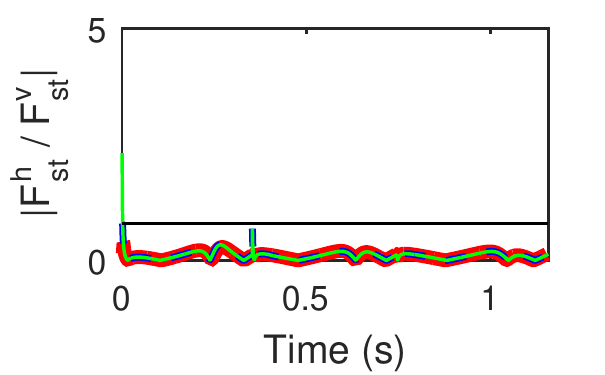}}}
	
	Case 1: $m_{load}=0 ~[kg]$
	
	{\resizebox{0.06\linewidth}{!}{\includegraphics{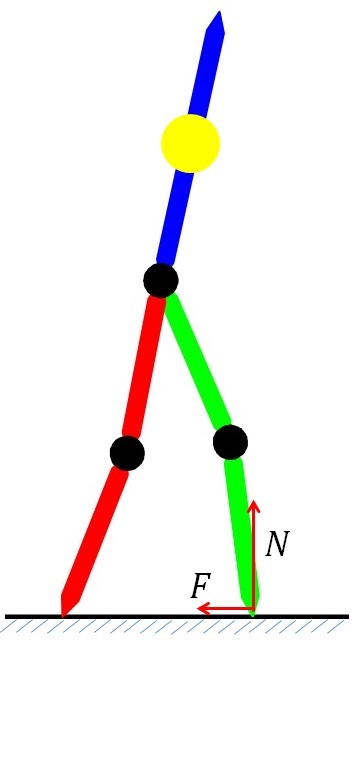}}}
	\subfloat
	{\resizebox{0.22\linewidth}{!}{\includegraphics{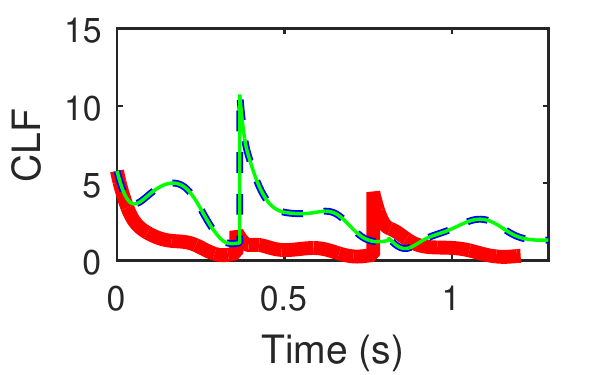}}}
	\subfloat
	{\resizebox{0.22\linewidth}{!}{\includegraphics{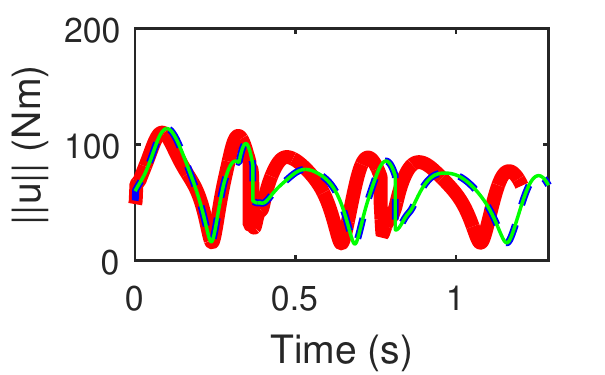}}}
	\subfloat
	{\resizebox{0.22\linewidth}{!}{\includegraphics{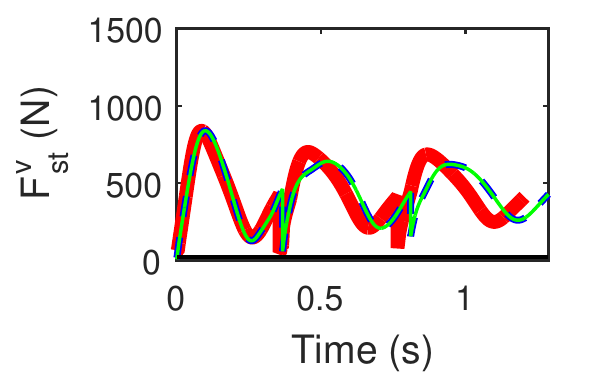}}}
	\subfloat
	{\resizebox{0.22\linewidth}{!}{\includegraphics{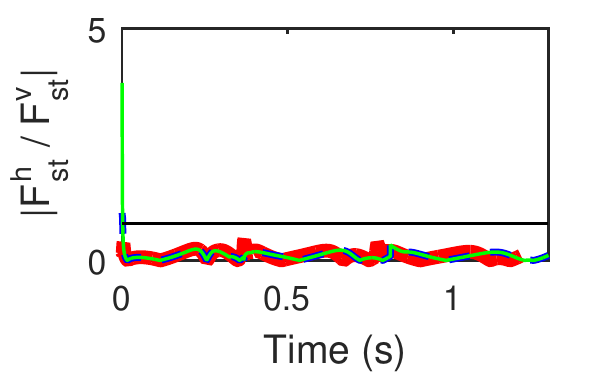}}}
	
	Case 2: $m_{load}=5 ~[kg]$
	
	{\resizebox{0.06\linewidth}{!}{\includegraphics{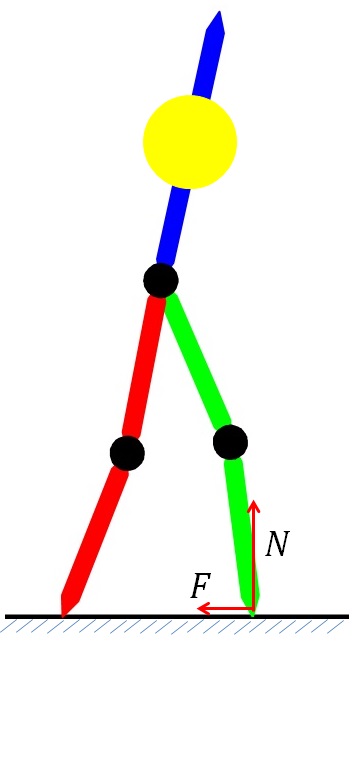}}}
	\subfloat
	{\resizebox{0.22\linewidth}{!}{\includegraphics{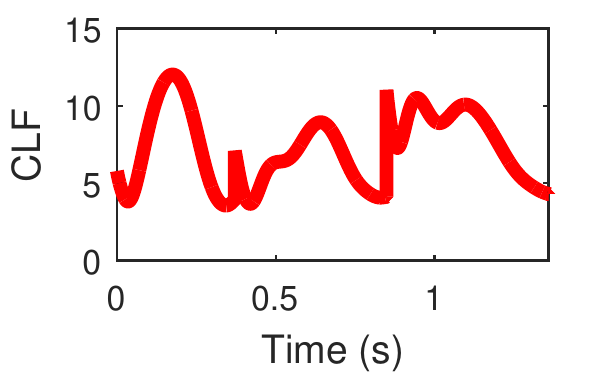}}}
	\subfloat
	{\resizebox{0.22\linewidth}{!}{\includegraphics{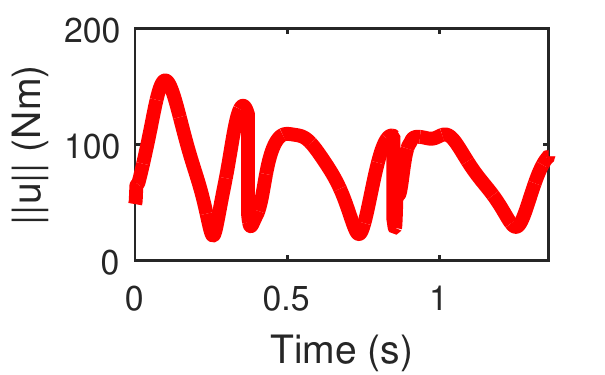}}}
	\subfloat
	{\resizebox{0.22\linewidth}{!}{\includegraphics{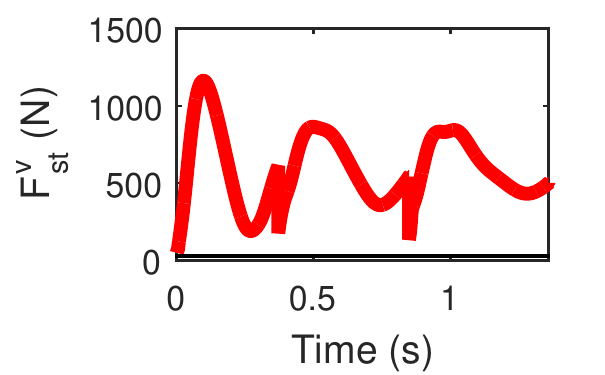}}}
	\subfloat
	{\resizebox{0.22\linewidth}{!}{\includegraphics{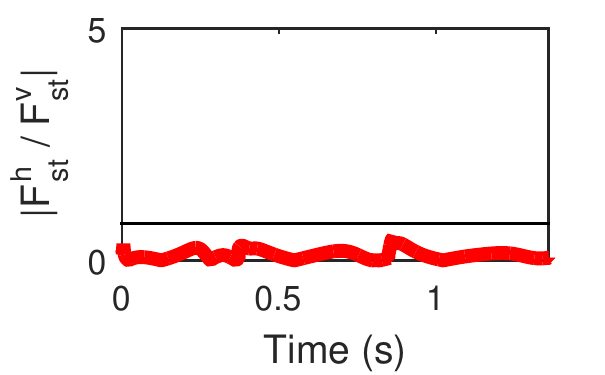}}}
	
	Case 3: $m_{load}=15 ~[kg]$	
	
	\subfloat
	{\resizebox{0.8\linewidth}{!}{\includegraphics[trim={0cm 5.7cm 0cm 0.6cm},clip]{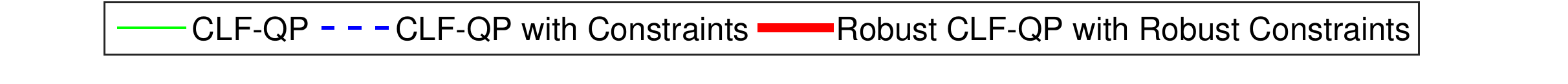}}}
	\caption{
		%
		\updateT{Three steps of walking of a b}ipedal robot \updateT{while} carrying \updateT{an} unknown load \updateT{with} contact force constraints. \updateT{The figure illustrates the tracking performance directly through the CLF (column 1), norm of the control inputs (column 2), non-negative vertical ground reaction force constraint (column 3) and friction cone constraints of being less than 0.8 (column 4).} Even in the nominal case of NO uncertainty, the CLF-QP controller fails \updateT{(the robot slips)} due to \updateT{violation} of \updateT{the} contact force constraints \updateT{(as seen in the rightmost figure of Case 1)}. The CLF-QP with Constraints (Contact Force Constraints) works well with perfect model but fails with only 5 kg of load \updateT{(as seen in the violation of the friction constraints in the rightmost figure of Case 2)}. The Robust CLF-QP with Robust Constraints maintains both good tracking performance and contact force constraints under up to 15 kg of load (47\% of the robot weight). The other two controllers are unstable in this case.
	}	
	\label{fig:load_friction}
	\vspace{-10pt}
\end{figure*}

\subsubsection{Dynamic Walking of Bipedal Robot while Carrying Unknown Load, subject to Contact Force Constraints}

The problem of contact force constraints is very important for robotic walking, as any violation of contact constraints would cause the robot to slip and fall. We design our nominal walking gait so as to satisfy contact constraints. However, we cannot guarantee this constraint once there is a perturbation from the nominal walking gait. Although the feedback controller (for example CLF-QP), will drive any error back to the periodic gait, however there is no way to enforce the contact force constraint. In our simulation, in addition to model uncertainty, we will introduce a small perturbation in initial configuration of the robot
, resulting in an initial CLF $V_0=\eta_0^T P \eta_0=5.8$ (see Fig.\eqref{fig:load_friction}). We will compare three different controllers,
\begin{align}
\begin{cases}
\text{I: CLF-QP}\\
\text{II: CLF-QP with Constraints (Contact Force Constraints)}\\
\text{III: Robust CLF-QP with Robust Constraints}
\end{cases}
\label{3ctrls_load_friction}
\end{align}

We also enforce input saturation constraints for all three controller in \eqref{3ctrls_load_friction}. However, since this constraint doesn't depend on the system model, a robust version for the constraint is not necessary.

Three simulation cases with different loads carried on the torso of the robot are conducted:
\begin{align}
\begin{cases}
\text{Case 1: $m_{load}= 0$ [kg]}\\
\text{Case 2: $m_{load}= 5$ [kg] (16\%)}\\
\text{Case 3: $m_{load}= 15$ [kg] (47\%)}
\end{cases}
\end{align}

We consider contact force constraints as follows. Let $F_{st}^h$ and $F_{st}^v$ be the horizontal and vertical contact force between the stance foot and the ground, in order to avoid slipping during walking, we will have to guarantee:
\begin{align}
F_{st}^v(x) \ge \delta_N > 0\\
\left| \frac{F_{st}^h(x)}{F_{st}^v(x)}\right|\le k_f
\end{align}
where $\delta_N$ is a positive threshold for the vertical contact force, and $k_f$ is the friction coefficient. In our simulation, we picked $\delta_N=0.1m_{robot}$ and $k_f=0.8$, with $m_{robot}=32 [kg]$ being the weight of the robot.

As we can see from Fig.\ref{fig:load_friction}, although we just generate a small initial perturbation, the controller I (CLF-QP) without considering contact force constraints still violated the friction constraint with $|F/N|_{max} \simeq 2.4$, while the controller II (CLF-QP with constraints) can handle this case well. However, with a small model uncertainty (adding $m_{load}=5[kg]$ to the torso), the controller B fails with the friction coefficient $|F/N|_{max} \simeq 1.1$.  Interestingly, in this case the robust CLF-QP with robust contact force constraints controller not only guarantees a very good friction constraints with $|F/N|_{max} \simeq 0.3$, but also has better tracking performance. With $m_{load}=15 [kg]$, while the two controllers I (CLF-QP) and II (CLF-QP with constraints) become unstable and fail in the first walking step, the controller III (Robust CLF-QP with robust constraints) still works well with $|k|_{max} \simeq 0.4$. Especially, we can notice from the figures of $\|u\|$ that the proposed robust CLF-QP with robust constraints has a much better performance in both cases, its range of control inputs is nearly the same with those of the rest two controllers.
In summary, we can conclude that the proposed robust QP offers a novel method that can increase the robustness of both stability and constraints while using the same range of control inputs with prior controllers. These properties will be further strengthened in the next interesting application for bipedal robotic walking with safety-critical constraints.

\subsubsection{Dynamic Walking of Bipedal Robot while Carrying Unknown Load, subject to Contact Force Constraints and Foot-Step Location Constraints}
For validating our proposed controller, we will also test with the problem of footstep location constraints when the robot carries an unknown load on the torso. The control methodology for this problem with perfect model can be found in \cite{CBF:Quan:Footstep:ADHS15}. We will run 100 simulations. For each simulation, the unknown load was chosen randomly between 5-15 kg, the desired footstep locations for 10 steps were chosen randomly between 0.35-0.55 m (the nominal walking gait has a step length of 0.45 m). Because the CLF-QP cannot handle footstep location constraints, the four following controllers will be compared:

\begin{align}
\begin{cases}
\text{I :CBF-CLF-QP (Foot-Step Placement)}\\
\text{II: CBF-CLF-QP with Constraints (Friction Constraints)}\\
\text{III: Robust CBF-CLF-QP}\\
\text{IV: Robust CBF-CLF-QP with Robust Constraints}
\end{cases}
\label{4_ctrls_footstep_friction}
\end{align}


\begin{figure}
	\centering
\resizebox{0.8\linewidth}{!}{\includegraphics{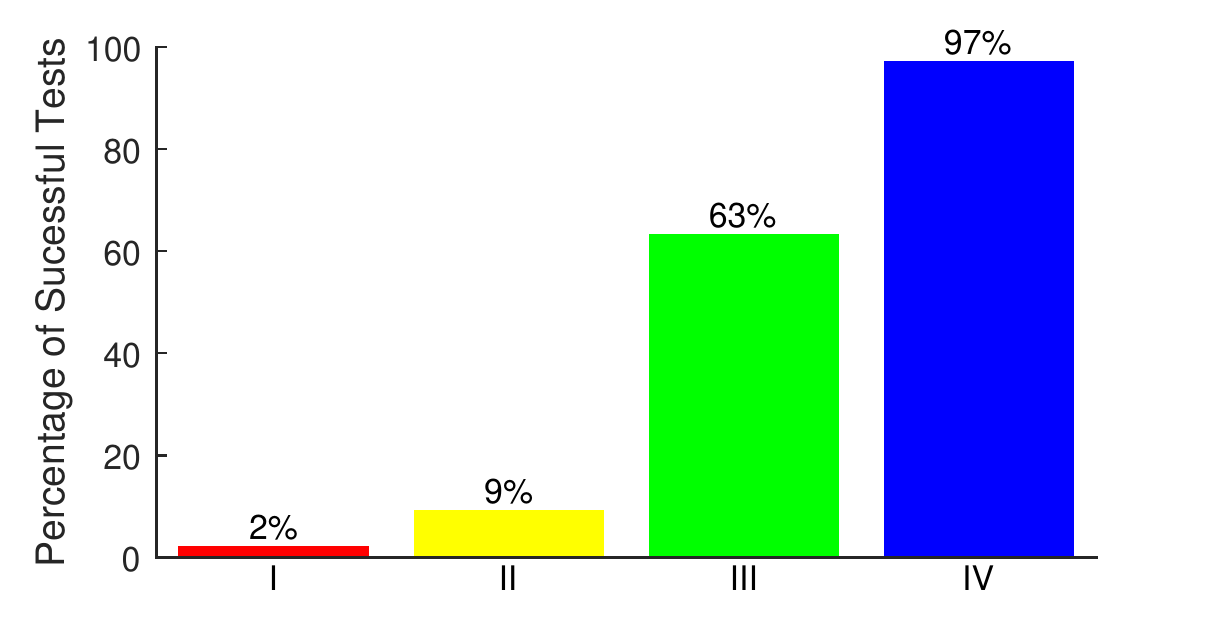}}
	\caption{Dynamic Walking of Bipedal Robot while Carrying Unknown Load, subject to Contact Force Constraints and Foot-Step Location Constraints. 100 random simulations were tested. For each simulation, the unknown load was choose randomly between 5-15 kg, the desired footstep locations for 10 steps were choose randomly between 0.35-0.55 m. The same set of random parameters was tested on the four controllers, where the four controller was specified in \eqref{4_ctrls_footstep_friction}.}	
	\label{fig:load_FSL}
	\vspace{-10pt}
\end{figure}

%
%


\begin{figure}
\centering
		\resizebox{1.0\linewidth}{!}{\includegraphics[trim={1cm 0.5cm 1cm 0.3cm},clip]{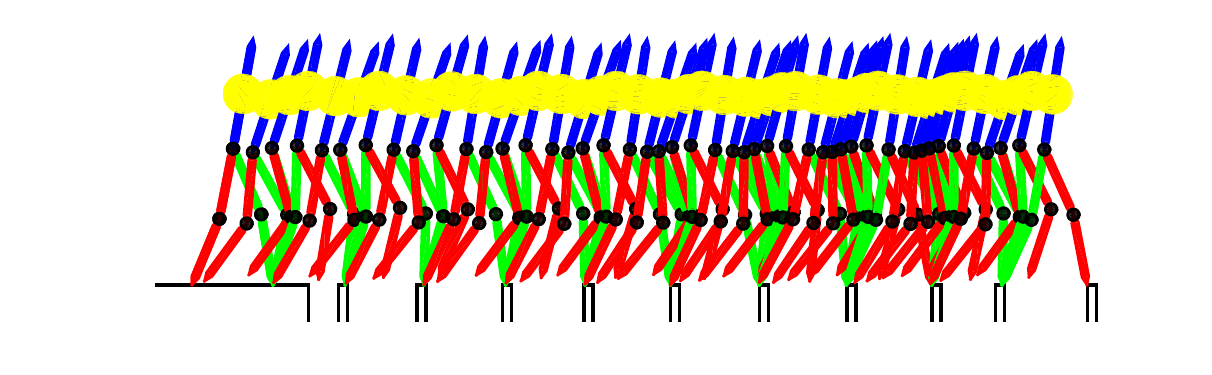}}
	\caption{Dynamic bipdal walking while carrying unknown load, subject to torque saturation constraints (input constraints), contact force constraints (state constraints), and foot-step location constraints (safety constraints). Simulation of the Robust CBF-CLF-QP with Robust Constraints controller for walking over 10 discrete foot holds is shown, subject to model uncertainty of 15 Kg (47 \%). A video of the simulation is available at \url{http://youtu.be/tT0xE1XlyDI}.}
	\label{fig:load_stone_stick_figure}
	\vspace{-10pt}
\end{figure}

\begin{figure}
	\centering
	\resizebox{1.0\linewidth}{!}{\includegraphics{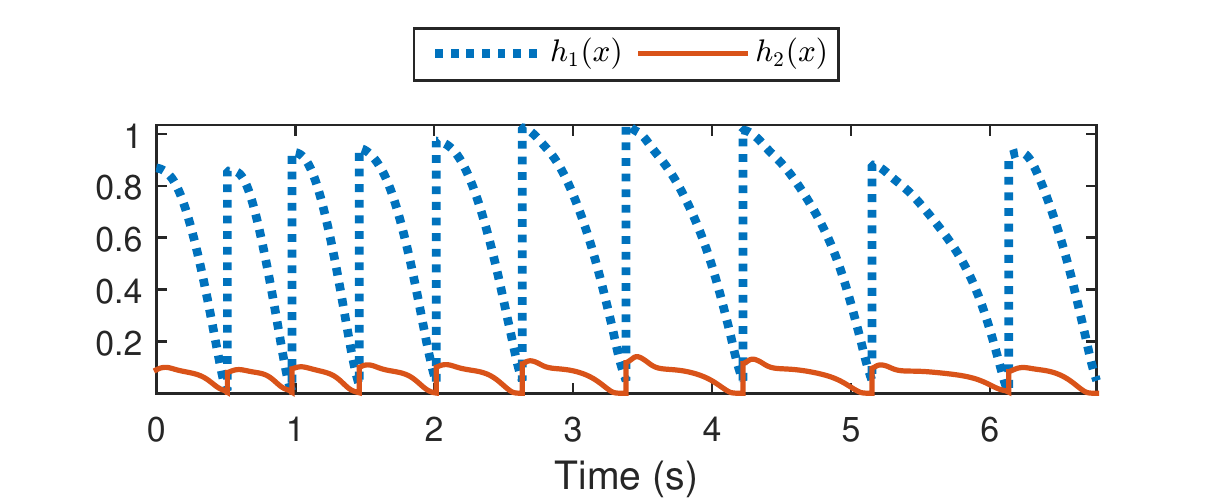}}
	\caption{Dynamic walking of bipedal robot while carrying unknown load of 15 Kg (47 \%). The CBF constraints, $h_1(x) \ge 0$ and $h_2(x) \ge 0$ defined in \cite{CBF:Quan:Footstep:ADHS15}, guarantee precise foot-step locations. The figure depicts data for 10 steps of walking.  As can be clearly seen, the constraints are strictly enforced despite the large model uncertainty.}
	\label{fig:load_stone_friction_CBF_CLF}
	\vspace{-10pt}
\end{figure}

\begin{figure}
	\centering
	\subfloat[][Vertical Contact Force: $N(x) > \delta_N, (\delta_N = 0.1 m g)$.]{\centering
		\resizebox{1.0\linewidth}{!}{\includegraphics{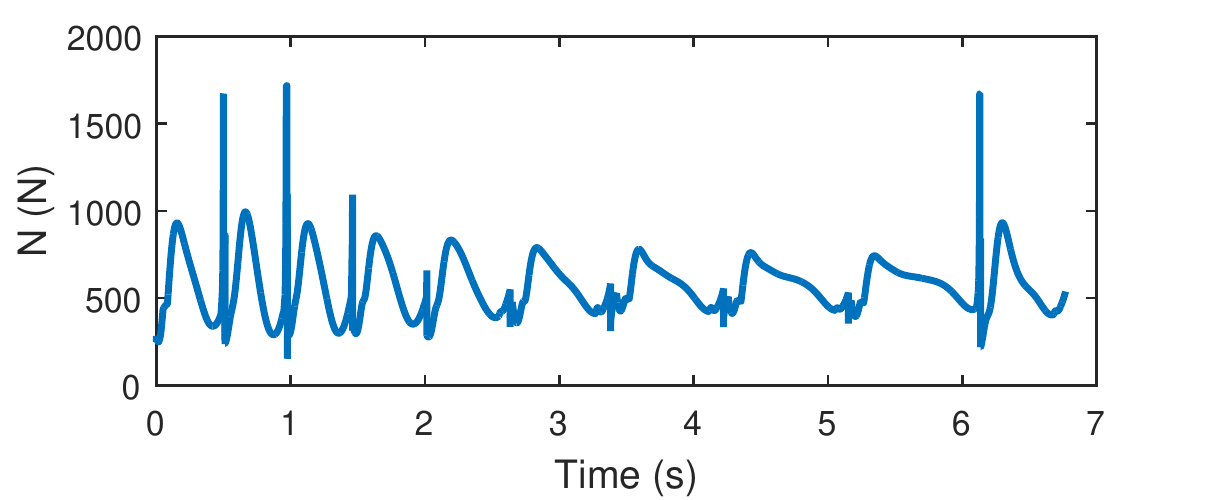}}
	}	\\
	\subfloat[][Friction Constraint: $|F/N|\le k_f$, ($k_f=0.8$)]{\centering
		\resizebox{1.0\linewidth}{!}{\includegraphics{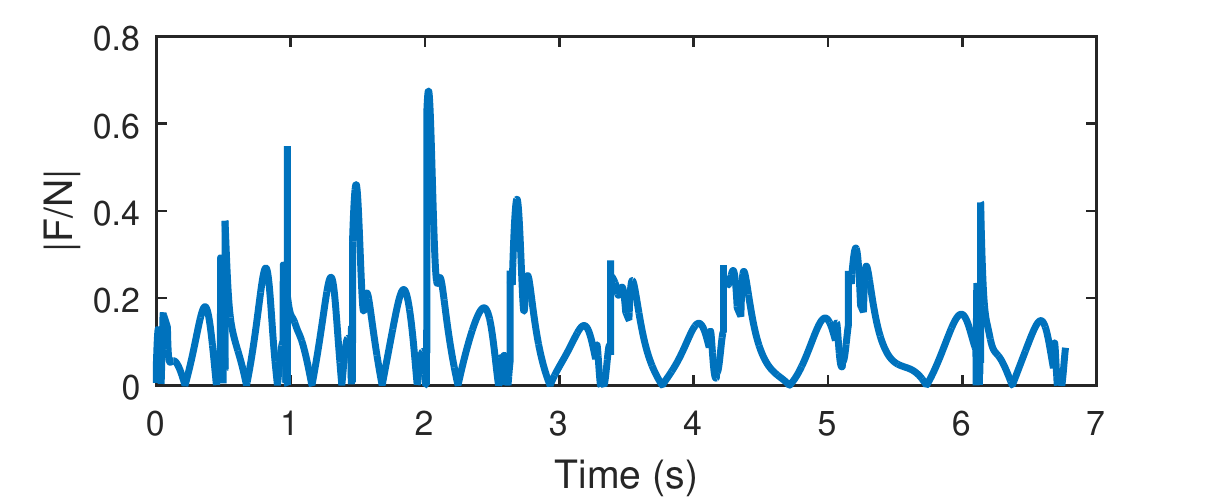}}
	}\\
	\caption{Dynamic walking of bipedal robot while carrying unknown load of 15 Kg (47 \%). (a) Vertical contact force constraint and (b) friction constraint are shown for 10 steps walking.  As is evident, both constraints are strictly enforced despite the large model uncertainty.}	
	\label{fig:load_stone_friction}
	\vspace{-10pt}
\end{figure}

As you can see from Fig.\ref{fig:load_FSL}, the performance of our proposed robust CBF-CLF-QP dominates that of the CBF-CLF-QP (97\% of success in comparision with 2\%). This result not only strengthens the effectiveness of the proposed controller, but it also emphasizes the importance of considering robust control for safety constraints because a small model uncertainty can cause violations of such constraints and thereby no longer guaranteeing safety.

Figures \ref{fig:load_stone_stick_figure}, \ref{fig:load_stone_friction_CBF_CLF}, \ref{fig:load_stone_friction},  illustrate one of the runs where the maximum load of 15 Kg (47\% of robot mass) was considered. Stick figure plot, CBF constraints, vertical contact force, and friction constraint plots are shown. Note that, the simulations were artificially limited to 10 steps, to enable fast execution of 100 runs for each controller. Simulations for larger number of steps were also successful as well, but are not presented here due to space constraints.  

\subsection{Experimental Results on Spring-Cart System}

\begin{figure}[]
	\centering
	\subfloat[Real system]{	\resizebox{0.7\linewidth}{!}{\includegraphics[trim={3cm 5cm 4cm 1cm},clip]{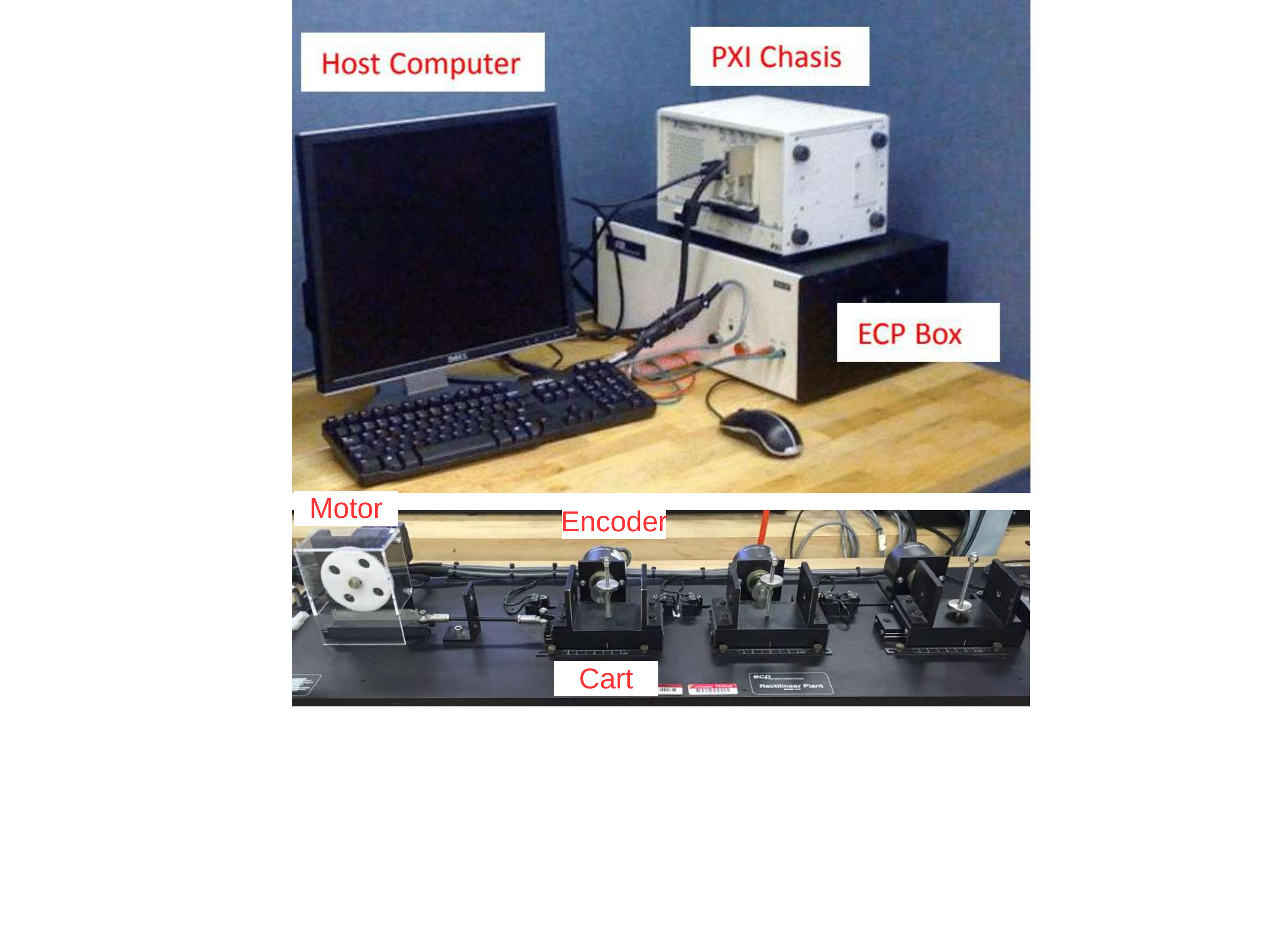}}}
	\\
	\subfloat[System Diagram]{	\resizebox{1.0\linewidth}{!}{\includegraphics[trim={1cm 12cm 1cm 0cm},clip]{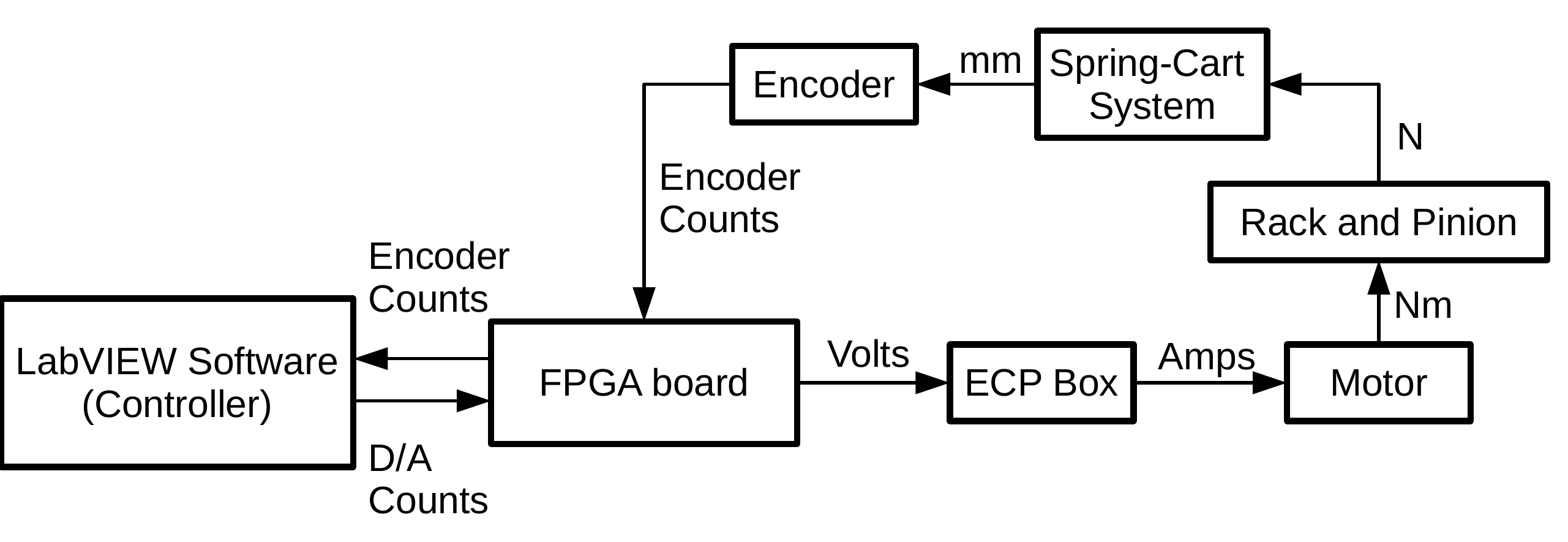}}}	
	\caption{Schematic of experiment setup for the spring-cart system.}
	\label{fig:spring_cart_system_diagram}
	\vspace{-10pt}
\end{figure}

\begin{figure*}
	\setlength{\columnseprule}{0.5pt}
	\begin{multicols}{3}
		\begin{figure}[H]
			\centering
			\subfloat{\centering
				\resizebox{0.4\linewidth}{!}{\includegraphics{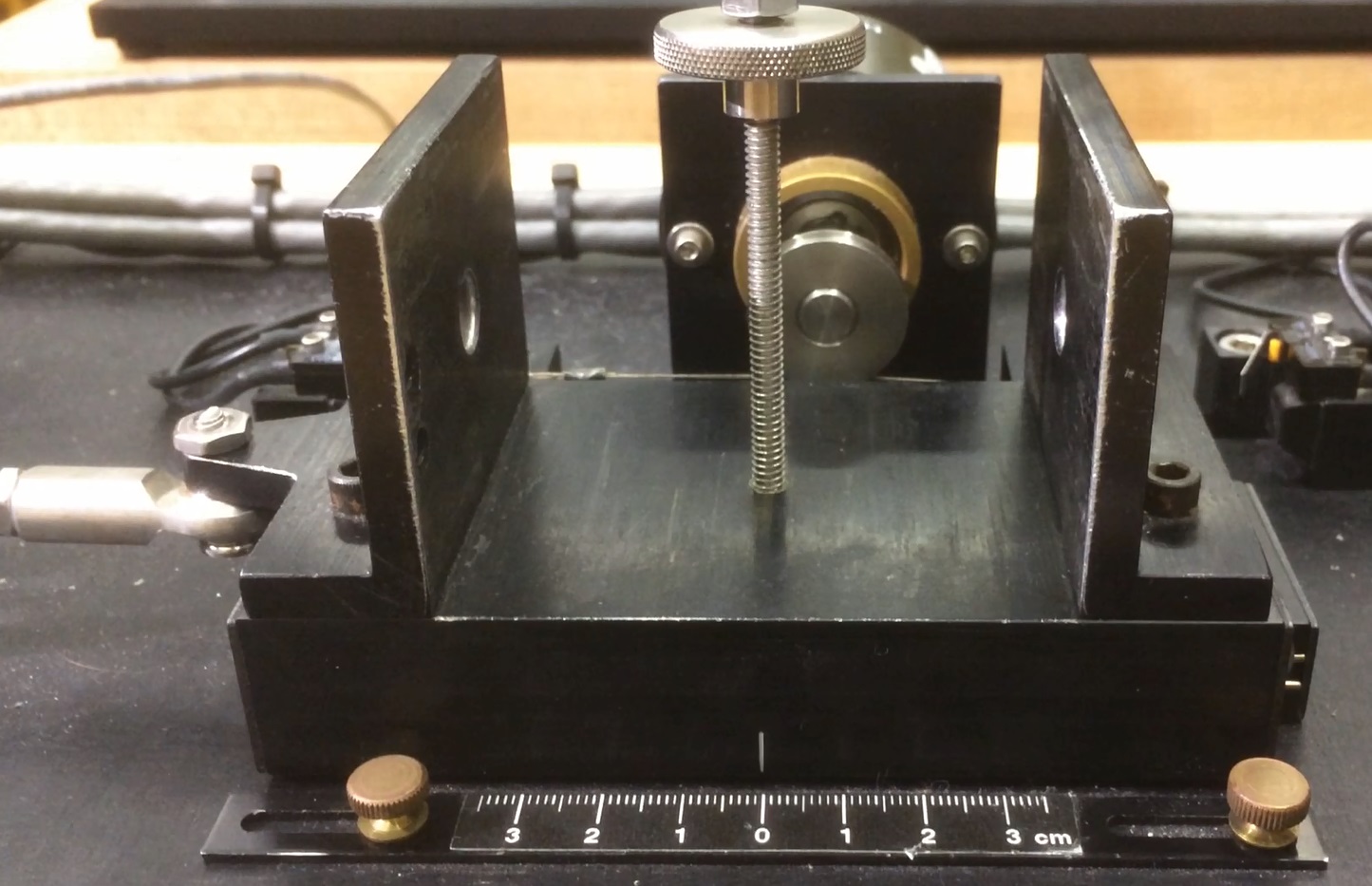}}
			}
			\subfloat{\centering
				\resizebox{0.6\linewidth}{!}{\includegraphics{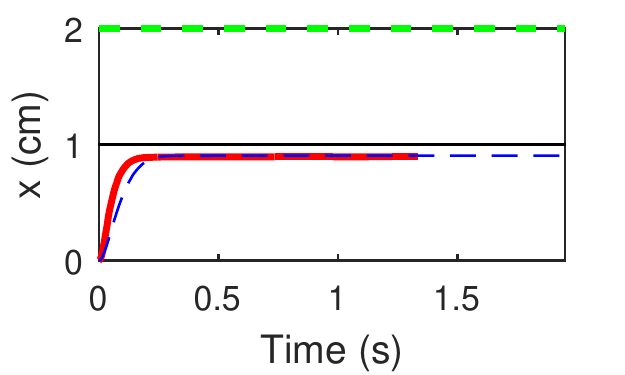}}
			}
			\\\small \linespread{1}\selectfont Case 1: No uncertainty\\
			
			\subfloat{\centering
				\resizebox{0.4\linewidth}{!}{\includegraphics{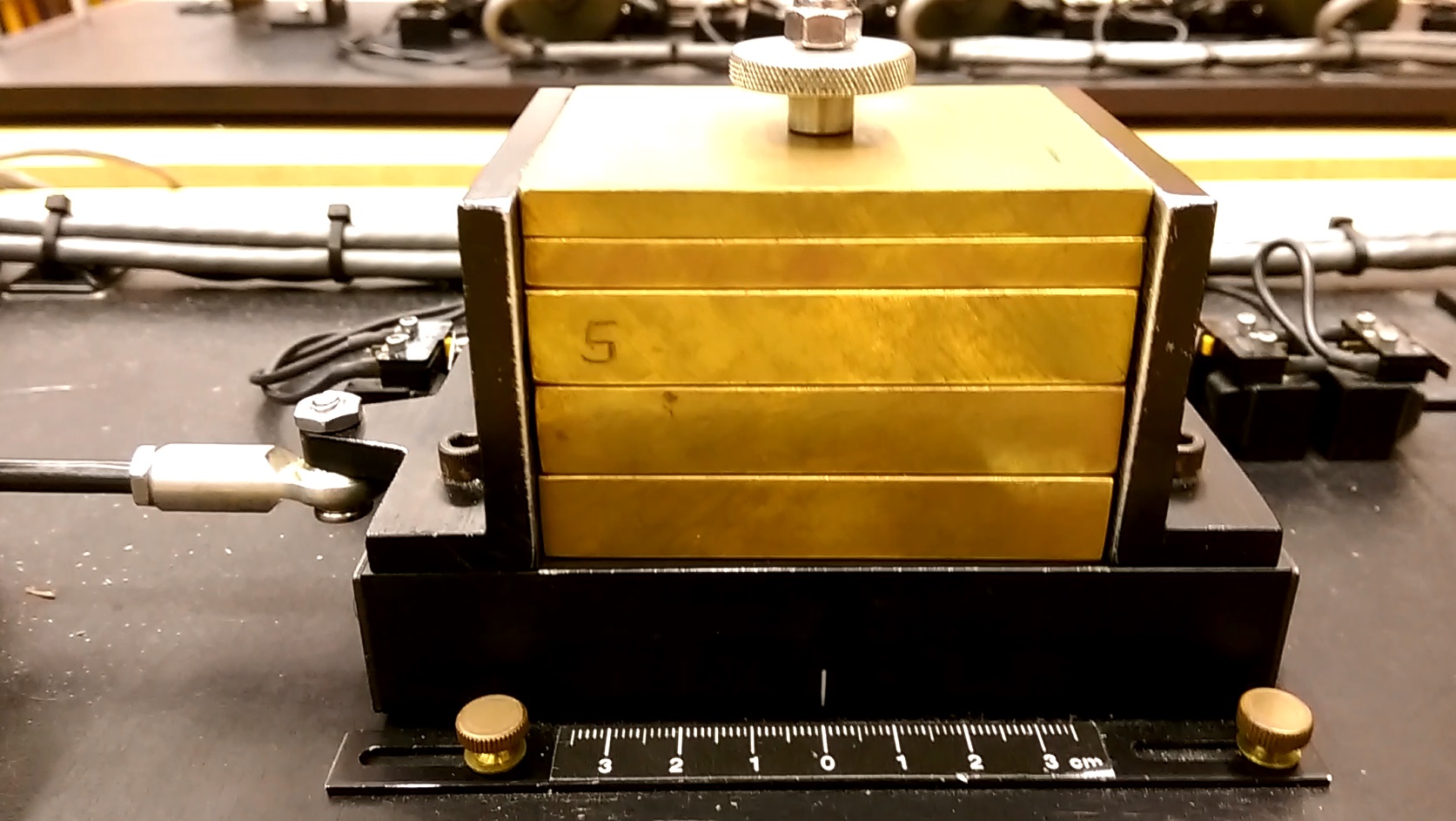}}
			}
			\subfloat{\centering
				\resizebox{0.6\linewidth}{!}{\includegraphics{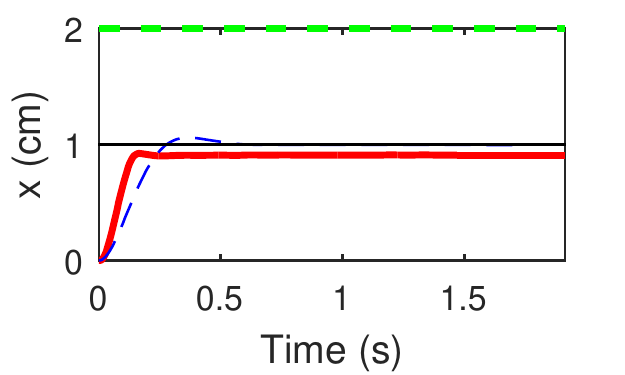}}
			}	
			\\Case 2: Full-loaded cart 1	
		\end{figure}
		
		\columnbreak
		
		\begin{figure}[H]
			\centering
						\subfloat{\centering
				\resizebox{0.4\linewidth}{!}{\includegraphics{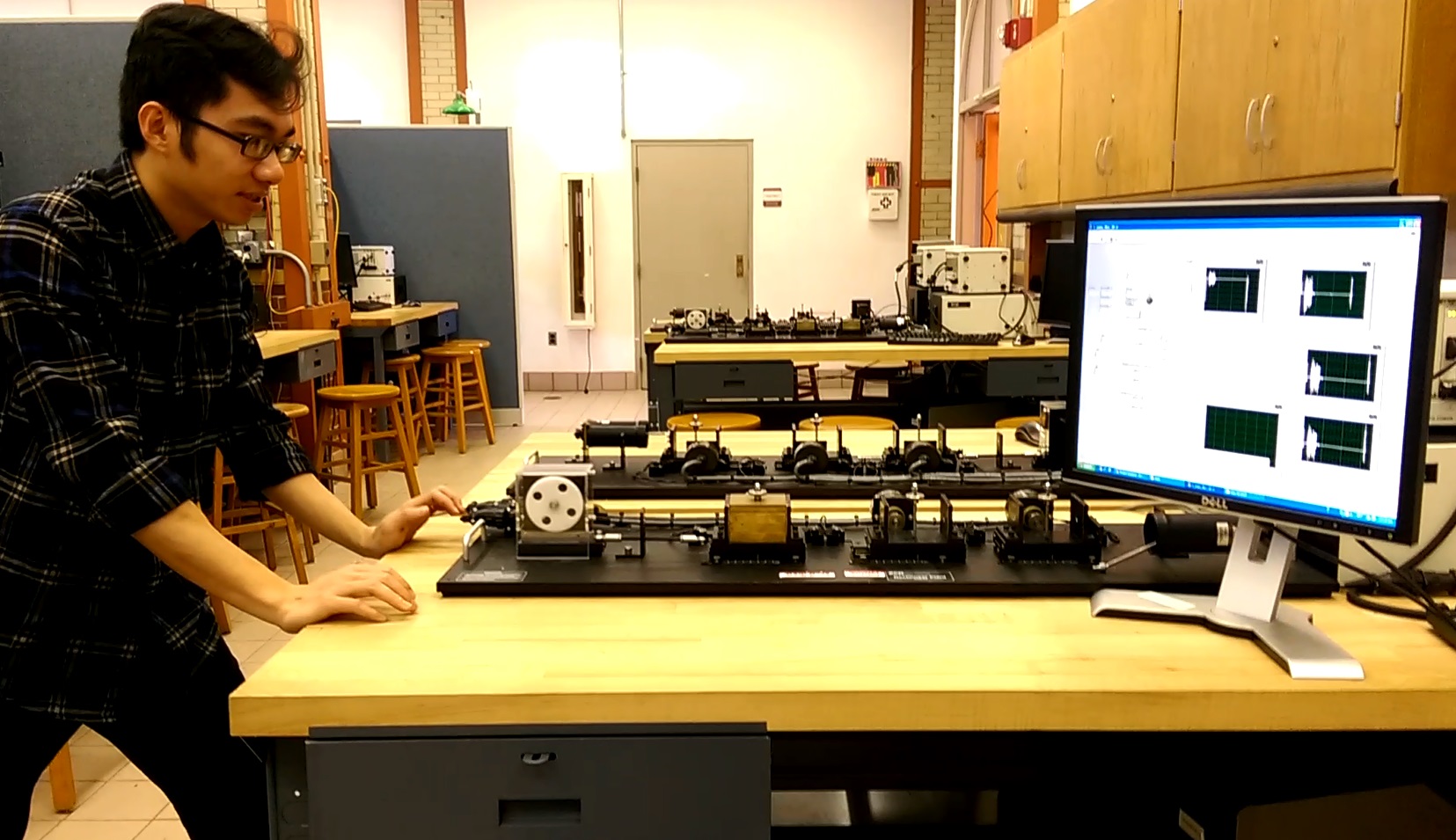}}
			}
			\subfloat{\centering
				\resizebox{0.6\linewidth}{!}{\includegraphics{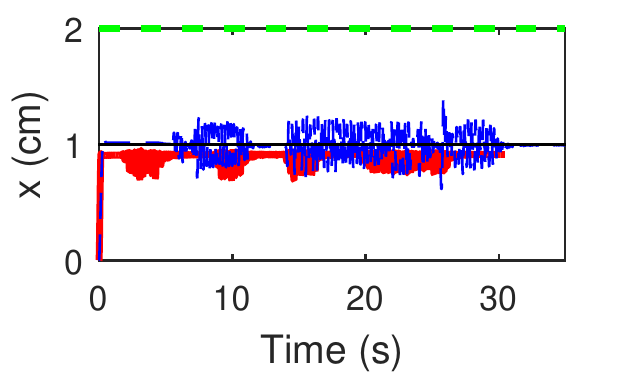}}
			}	
			\\Case 3: Case 2 and shaking the table\\
			
			\subfloat{\centering
				\resizebox{0.4\linewidth}{!}{\includegraphics{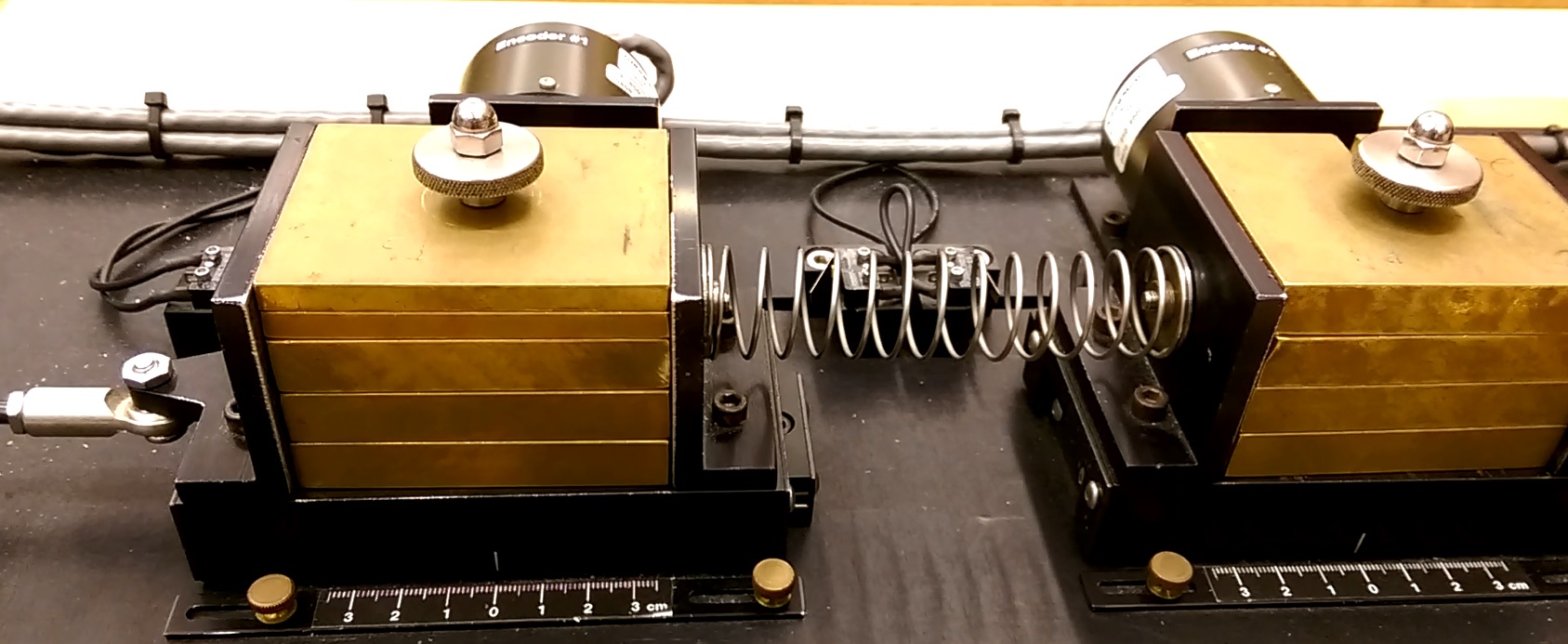}}
			}
			\subfloat{\centering
				\resizebox{0.6\linewidth}{!}{\includegraphics{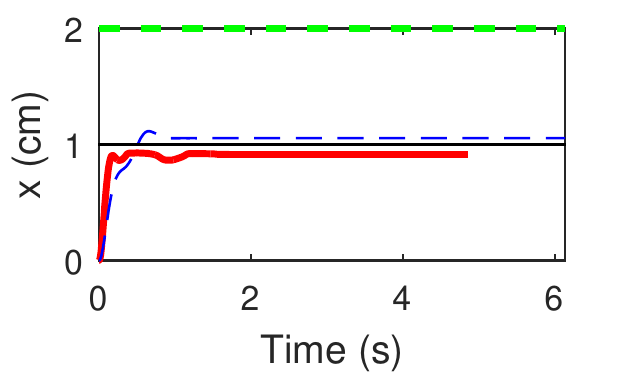}}
			}	
			\\ \small \linespread{1}\selectfont Case 4: Full-loaded cart 1 + spring + full-loaded cart 2
		\end{figure}
	
		\columnbreak

\begin{figure}[H]
	\centering
	\subfloat{\centering
		\resizebox{0.4\linewidth}{!}{\includegraphics{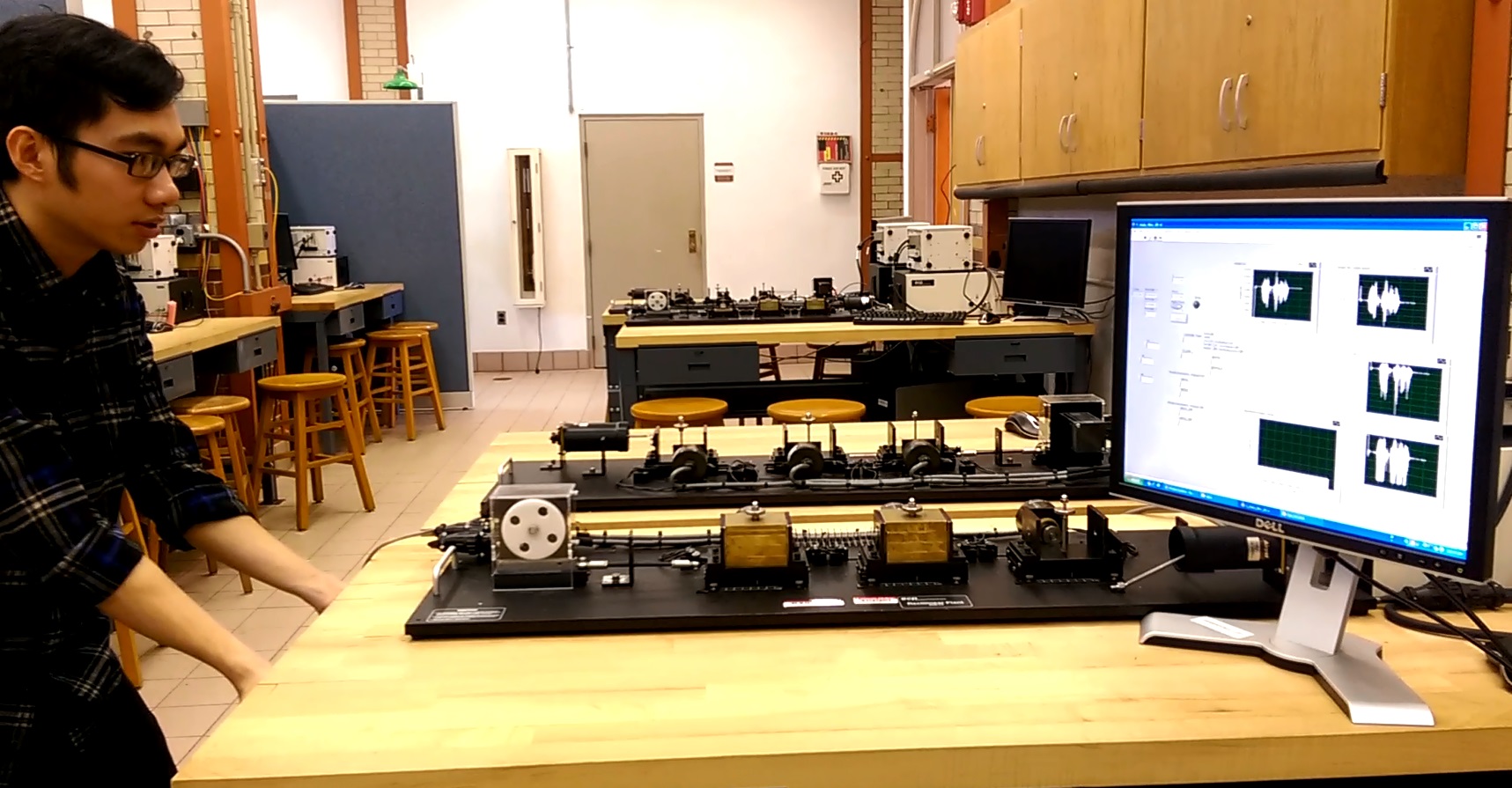}}
	}
	\subfloat{\centering
		\resizebox{0.6\linewidth}{!}{\includegraphics{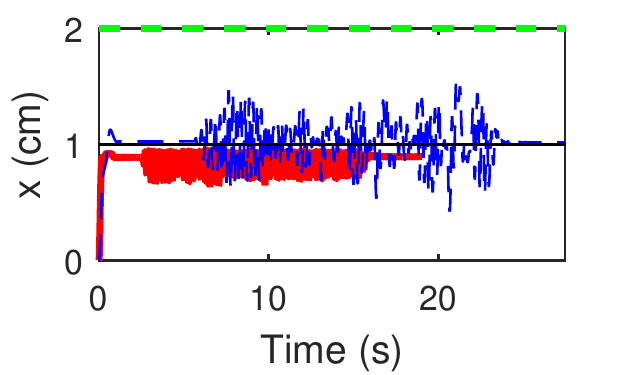}}
	}	
	\\Case 5: Case 4 and shaking the table\\
	
	\subfloat{\centering
		\resizebox{0.4\linewidth}{!}{\includegraphics{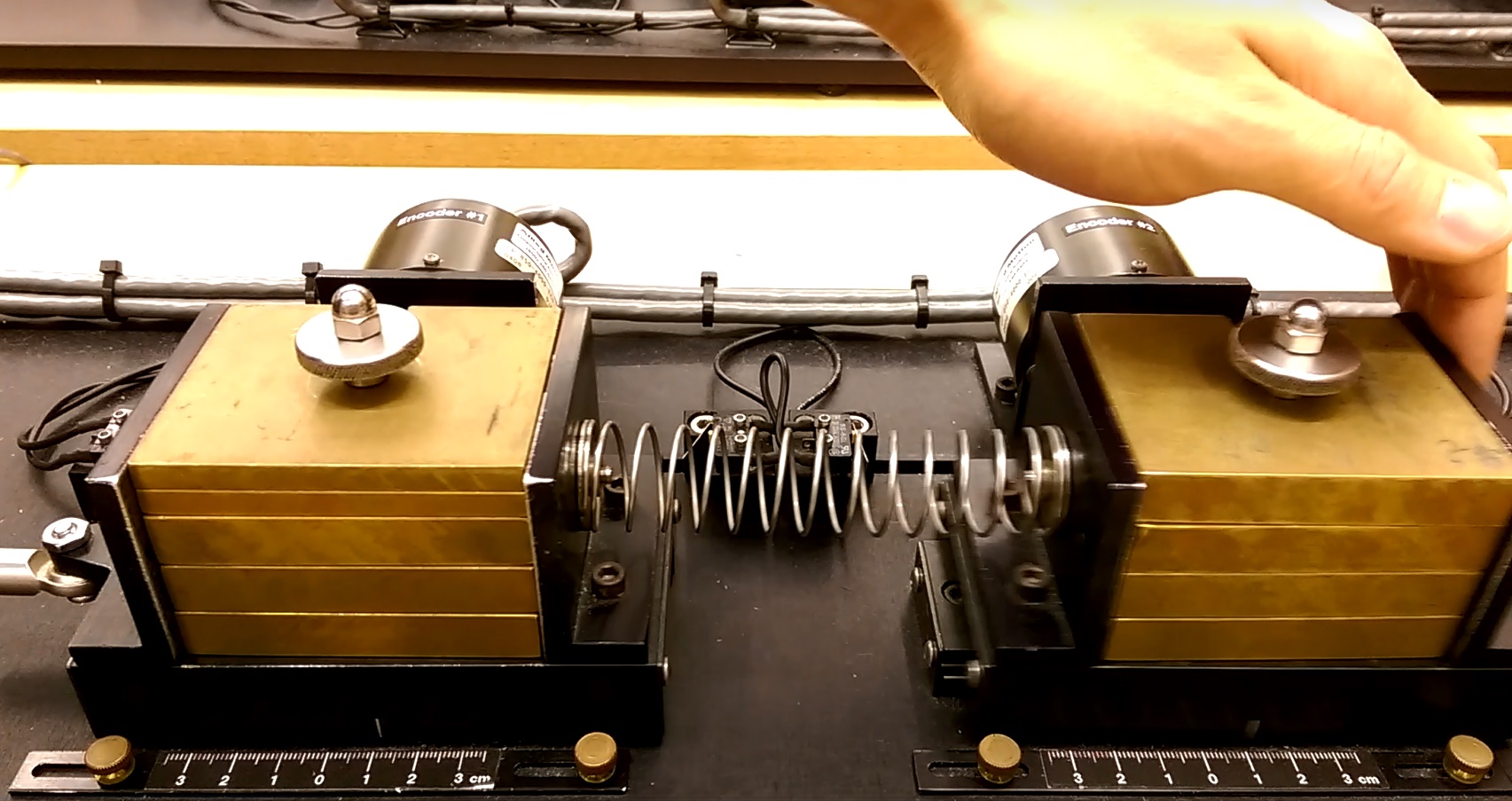}}
	}
	\subfloat{\centering
		\resizebox{0.6\linewidth}{!}{\includegraphics{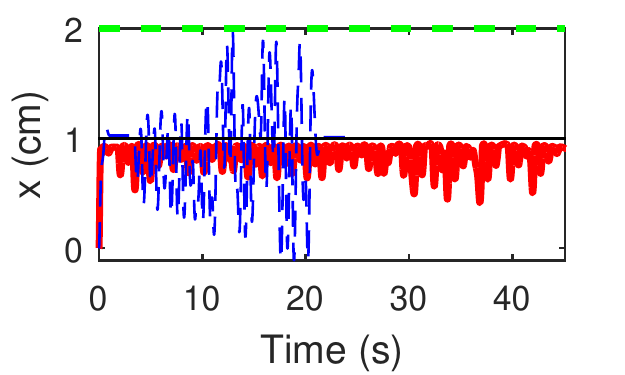}}
	}	
	\\Case 6: Case 4 and perturbing the cart 2.
\end{figure}

	\end{multicols}
	\centering
	\subfloat
	{\resizebox{1\linewidth}{!}{\includegraphics[trim={0cm 4cm 0cm 0.2cm},clip]{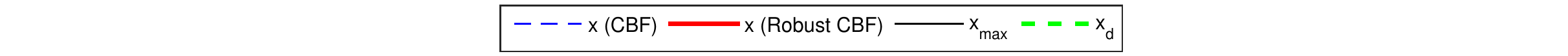}}}
	\caption{Experimental results on the spring-mass system. The goal is to drive the cart to the target location of 2 cm, while enforcing the safety constraint that the cart does not cross 1 cm.  The controller just uses the nominal model as illustrated in Case 1 for all the 6 different cases.  Model uncertainty is introduced for Cases 2 in the form of an added unknown mass.  Additionally for Case 3, a perturbation is introduced in the form of shaking the table.  For Case 4, in addition to the unknown mass, an unkown dynamics is introduced in the form of another cart that is connected through a spring.  Additionally for Cases 5 and 6, perturbations are introduced in the form of manually shaking the table and shaking the second cart respectively.  In all cases, the proposed robust controller still enforces the strict safety-critical constraint and maintains the cart position under 1 cm.  A video of the experiment is available at \url{https://youtu.be/g1UewP4R8L4}.}
	\label{fig:experiment}	
	\vspace{-10pt}
\end{figure*}

\editQTN{Having presented numerical results of our proposed control method for bipedal robots, we next present e}xperimental validation for the method on a rectilinear spring-cart system, \editQTN{as} shown in Fig.~\ref{fig:experiment}. \editQTN{It must be noted that the experiments on a rectilinear spring-cart system are connected to the simulations with a bipedal robot.  Since our simulations consider nonlinear systems with IO linearization controllers and safety-constraints with relative-degree two, our preliminary experiments are thus with linear systems with relative-degree two safety-constraints. Furthermore, with this experimental result, we will perturb the system with different types of model uncertainties (see Fig. \ref{fig:experiment}). It therefore can represent model uncertainty in the IO linearized system discussed in Section \ref{sec:robust_control}.  Future work will consider experiments with bipedal systems}.

For this experiment, our control problem is to track desired set-point ($x \to x_d=2~[cm]$) by using \editQTN{a} CLF, and guarantee state-dependent constraint ($x \le x_{max}=1~[cm]$) by using \editQTN{a} CBF, where $x$ is the position of \editQTN{C}art 1.

\editQTN{The experimental setup of the rectilinear spring-cart system (ECP Model 210) is described in Fig \ref{fig:spring_cart_system_diagram}. 
Our controller is implemented using LabVIEW, wherein we call a custom-generated C++ code that implements a fast QP solver. This QP solver code is autogenerated using CVXGEN \cite{CVXGEN:MaBo2012}.  The controller runs at 100 Hz on a LabVIEW PXI DAQ and outputs the control input to the ECP system.
In particular, the control input is sent to a FPGA  board in the PXI DAQ, which then generates  and outputs	an analog driving voltage (through a 16-bit Digital-to-Analog Converter) to the current amplifier in the ECP system. This amplifier generates the required current to drive the motor which in turn produces a torque. For rectilinear systems this torque is converted to a linear force through a rack and pinion mechanism.
 The motion of the moving cart is  measured  by an  encoder and this information in  encoder counts is  acquired by the FPGA board in the PXI DAQ and sent to our controller via the host LabVIEW software.} 

Fig.\ref{fig:experiment} compares the performance of two controllers CBF-CLF-QP (dotted blue line) and Robust CBF-CLF-QP (red line) under six different cases. Experimental setup for each case and corresponding result are shown in Fig.\ref{fig:experiment}. Note that the two controllers were designed based on the nominal model indicated in case 1 (a single cart) and we generated model uncertainty  from case 2-6 by adding masses, shaking table, adding spring and another cart, etc.

From Fig.\ref{fig:experiment}, we can observe clearly that while in case 1 (without model uncertainty), the two controllers have almost the same performance, from Cases 2-6, our proposed robust CBF-CLF-QP outperformed the nominal CBF-CLF-QP. To be more specific, the robust CBF-CLF-QP controller maintains the constraint ($x \le 1 (cm)$) very well, the nominal CBF-CLF-QP fails in all last 5 cases.

\subsection{Discussion}
The proposed controller has a few shortcomings. 
Since we are solving for the control input under the worst-case (bounded) model uncertainty assumption, the control could be aggressive.  This is a typical drawback of robust controllers.
Moreover, as mentioned in Remark \ref{remark:local-feasibility}, we only have local feasibility of the QP.  In particular, if we increase the bounds of the uncertainty significantly, i.e., large values of $\Delta_1^{max}, \Delta_2^{max}$, the optimization solves for the control input for the worst case, and could potentially lead to infeasibility of the QP. Thus, there is a trade-off between robustness and feasibility of the controller. If we choose a small bound on the model uncertainty, it could lead to poor tracking stability and potential violation of the safety constraint under mild model uncertainty that exceeds the bounds. In contrast, if we assume the bound of model uncertainty being too large the QP could become infeasible. Therefore, in the future, a more formal design of the bounded uncertainty assumption should be explored. 

\section{Conclusion}
\label{sec:conclusion}
We have presented a novel Optimal Robust Control technique that uses control Lyapunov and barrier functions to successfully handle significantly high model uncertainty for both stability, input-based constraints, state-dependent constraints, and safety-critical constraints. We validated our proposed controller on different problems both numerically and experimentally, which show the same property: under model uncertainty, our Robust Control based QP, has much better tracking performance and guarantees desired constraints while other types of QP controllers using Lyapunov and barrier functions not only have large tracking errors but also violate the constraints with even a small model uncertainty. We show numerical validations on dynamically walking of bipedal robots while subject to torque saturation and contact force constraints in the presence of model uncertainty, and dynamically walking with precise foot placements over a terrain of stepping stones while subject to model uncertainty.  We also experimentally validate our controller on a spring-cart system subject to significant model uncertainty and perturbations.
Future directions involve experimental validations on bipedal robots and other dynamic robotic systems.



\appendices
\section{Proof of Theorem \ref{thm:MainRelaxedIneq}}
\label{sec:proof_2_appendix}
In this Appendix, we will present a detailed proof of Theorem \ref{thm:MainRelaxedIneq} about the stability of CLF based controller with relaxed RES-CLF condition for hybrid systems.

A large part of this proof directly follows from results and proofs in \cite{RESCLF:AmGaGrSr:TAC12}, that is used to prove the stability of the hybrid system under the RES-CLF condition. In our case of relaxed CLF, we will state the additional condition under which the proof is still valid. 

Let $\epsilon >0$ be fixed and select a Lipschitz continuous feedback $u_\epsilon$ of the relaxed CLF-QP controller \eqref{eq:RelaxedCLFQP}. From \cite[(56)]{RESCLF:AmGaGrSr:TAC12}, we have $T_{I}^\epsilon(\eta,z)$ is continuous (since it is Lipschitz) and therefore there exists $\delta >0$ and $\Delta T >0$ such that for all $(\eta,z) \in B_{\delta}(0,0) \cap \S$,
\begin{align}
\label{eqn:TstarTIBound}
T^*-\Delta T \le T_{I}^\epsilon(\eta,z) \le T^*+\Delta T,	
\end{align}
where $T^*$ is the period of the orbit $\Orbitz$.

In order to make use of the proof of the exponential stability for the standard RES-CLF controller in \cite{RESCLF:AmGaGrSr:TAC12}, we will present the condition for bounding the system states $\eta(t)$ at the time-to-impact $\TI$  in the following lemma.

\gap
\begin{lemma}
\label{lem:BoundXt}
{\it Let $\Orbitz$ be an exponentially stable periodic orbit of the hybrid zero dynamics $\HS|_Z$ \eqref{eq:hybrid_zero_dyn} transverse to $\S \cap Z$ and the continuous dynamics of ${\mathscr H}$ \eqref{eq:lineaized_sys} controlled by a CLF-QP with relaxed inequality \eqref{eq:RelaxedCLFQP}. Then for each $\Delta T >0$ and $\epsilon > 0$, there exists an $\bar{w}_{\epsilon} \ge 0$ such that, if the solution $u_\epsilon(\eta,z)$ of the CLF-QP with relaxed inequality \eqref{eq:RelaxedCLFQP} satisfies $w_{\epsilon}(\TI) \le \bar{w}_{\epsilon}$, then
\begin{align}
\label{BoundXtX0Lemma1}
\|\eta(t) \| \biggr|_{t=\TI} &\le \sqrt{\frac{c_2}{c_1}} \frac{2e^{-1}}{(T^*-\Delta T)c_3}\|\eta(0) \|.
\end{align}
}
\end{lemma}
\gap

\begin{proof}
From \eqref{XtX0} and because $w_{\epsilon}(\TI) \le \bar{w}_{\epsilon}$, it implies that 

\begin{align}
\|\eta(t) \| \biggr|_{t=\TI} &\le \sqrt{\frac{c_2}{c_1}} \frac{1}{\epsilon} e^{-\frac{c_3}{2\epsilon}\TI+\frac{1}{2}w_{\epsilon}(\TI)}  \|\eta(0) \| \nonumber \\
&\le \sqrt{\frac{c_2}{c_1}} \frac{1}{\epsilon} e^{-\frac{c_3}{2\epsilon}\TI+\frac{1}{2}\bar{w}_{\epsilon}}  \|\eta(0) \| 
\label{bound xtx0 bar c3}
\end{align}
Then, from \eqref{eqn:TstarTIBound}, we have,
\begin{align}
\label{eq:bound_eta_TI}
\|\eta(t) \| \biggr|_{t=\TI} \le \sqrt{\frac{c_2}{c_1}} \frac{1}{\epsilon} e^{-\frac{c_3}{2\epsilon}(T^*-\Delta T)+\frac{1}{2}\bar{w}_{\epsilon}}  \|\eta(0) \|.
\end{align}

Furthermore, because $e^{-\alpha}\le e^{-1}/\alpha, \forall \alpha \ge 0$, we have:
\begin{align}
\label{ieq: c3 condition}
\frac{1}{\epsilon}e^{-\frac{c_3}{2\epsilon}(T^*-\Delta T)} \le \frac{2e^{-1}}{(T^*-\Delta T)c_3}.
\end{align}

Then it implies that there exists a $\bar{c}_3 \le c_3$ such that:
\begin{align}
\label{bound exp bar c3}
\frac{1}{\epsilon}e^{-\frac{c_3}{2\epsilon}(T^*-\Delta T)} \le \frac{1}{\epsilon}e^{-\frac{\bar{c}_3}{2\epsilon}(T^*-\Delta T)} \le \frac{2e^{-1}}{(T^*-\Delta T)c_3}.
\end{align}
The above inequality follows by the fact that $\frac{1}{\epsilon}e^{-\frac{c_3}{2\epsilon}(T^*-\Delta T)}$ is a monotonically decreasing function of $c_3$, and that for any real numbers $a \le b$, there always exists a number $c \in [a:b]$ such that $a\le c\le b$. To be more specific, let $\bar c_3^a$ be the solution of:
\begin{align}
\label{def c31 bar}
 \frac{1}{\epsilon}e^{-\frac{\bar{c}_3^a}{2\epsilon}(T^*-\Delta T)} = \frac{2e^{-1}}{(T^*-\Delta T)c_3},
\end{align}
then any $\bar c_3 \in [\bar c_3^a:c_3]$ will satisfy \eqref{bound exp bar c3}.

We can then define
\begin{align}
 \label{bar_we_def}
& \bar{w}_{\epsilon}:=\frac{c_3-\bar{c}_3}{\epsilon}(T^*-\Delta T) \ge 0,\\
 \label{eq:c3_wbar_c3bar}
\Rightarrow   &  -\frac{c_3}{2\epsilon}(T^*-\Delta T)+\frac{1}{2}\bar{w}_{\epsilon}=-\frac{\bar{c}_3}{2\epsilon}(T^*-\Delta T).
\end{align}

Next, plugging \eqref{eq:c3_wbar_c3bar} into \eqref{eq:bound_eta_TI}, we have,
\begin{align}
\label{relaxed RES condition x}
\|\eta(t) \| \biggr|_{t=\TI} &\le \sqrt{\frac{c_2}{c_1}} \frac{1}{\epsilon} e^{-\frac{\bar{c}_3}{2\epsilon}(T^*-\Delta T)}  \|\eta(0) \|, \\
\label{relaxed RES condit
	ion Ve}
V_{\epsilon}(\eta(t))\bigr|_{t=\TI} &\le e^{-\frac{\bar{c}_3}{\epsilon}(T^*-\Delta T)}V_{\epsilon}(\eta(0)).
\end{align}

We now complete the proof of Lemma \ref{lem:BoundXt} by substituting \eqref{bound exp bar c3} into \eqref{relaxed RES condition x} to obtain \eqref{BoundXtX0Lemma1}. The solution of $\bar w_{\epsilon}$ can be found in \eqref{bar_we_def}.
\end{proof}

Using Lemma \ref{lem:BoundXt}, we then can follow the same protocol of the proof in \cite[Theorem 2]{RESCLF:AmGaGrSr:TAC12} until equation \cite[(64)]{RESCLF:AmGaGrSr:TAC12}.

We then define $\beta_1(\epsilon) = \frac{c_2}{\epsilon^2} L_{\Rx}^2 e^{-\frac{c_3}{\epsilon} (T^*-\Delta T)} $ and $\bar{\beta}_1(\epsilon) = \frac{c_2}{\epsilon^2} L_{\Rx}^2 e^{-\frac{\bar{c}_3}{\epsilon} (T^*-\Delta T)} $ where $L_{\Rx}$, defined after \cite[(59)]{RESCLF:AmGaGrSr:TAC12}, is the Lipschitz constant for $\Rx$.
Because $\beta_1(0^+) := \lim_{\epsilon \searrow 0} \beta_1(\epsilon) = 0$, then there exists an $\overline{\epsilon}$ such that
\begin{eqnarray}
\beta_1(\epsilon)  < c_1 \qquad \qquad \forall \qquad 0 < \epsilon < \overline{\epsilon}.
\end{eqnarray}
and for each  $\epsilon$, 
if we define $\bar c_3^b$ be the solution of
\begin{align}
\label{def c32 bar}
\frac{c_2}{\epsilon^2} L_{\Rx}^2 e^{-\frac{\bar{c}_3^b}{\epsilon} (T^*-\Delta T)}=c_1, 
\end{align}
then
for $\bar c_3 \in (\bar c_3^b : c_3]$, we have,
\begin{eqnarray}
\label{barbeta1bound}
\beta_1(\epsilon) \le \bar{\beta}_1(\epsilon) < c_1.
\end{eqnarray} 
However, as presented in proof of Lemma \ref{lem:BoundXt}, $\bar c_3$ also needs to satisfy \eqref{bound exp bar c3}. Therefore, in order to guarantee the satisfaction of both \eqref{bound exp bar c3} and \eqref{barbeta1bound}, $\bar c_3$ needs to be chosen in the following set
\begin{align}
\label{c3 bar set}
\bar c_3 \in \{[\bar c_3^a : c_3] \cap (\bar c_3^b : c_3]\},
\end{align} 
where $\bar c_3^a$ and $\bar c_3^b$ are defined in \eqref{def c31 bar} and \eqref{def c32 bar} respectively.

The rest of the proof follows from the proof of \cite[Theorem 2]{RESCLF:AmGaGrSr:TAC12} using $\bar \beta_1$ instead of $\beta_1$. We finish our proof with the value of $\bar w_{\epsilon}$ determined via \eqref{bar_we_def}, in which the feasible set of $\bar c_3$ is defined in \eqref{c3 bar set}.


\bibliographystyle{IEEEtran}
\bibliography{Quan_ref}

\begin{IEEEbiography}[{
		\includegraphics[width=1in,height=1.25in,clip,keepaspectratio]{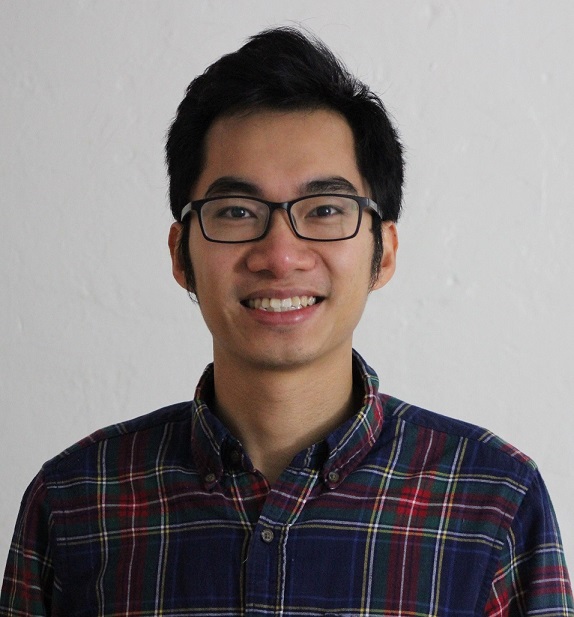}
	}]{Quan Nguyen}
is an Assistant Professor of Aerospace and Mechanical Engineering at the University of Southern California. Prior to joining USC, he was a Postdoctoral Associate in the Biomimetic Robotics Lab at the Massachusetts Institute of Technology (MIT). He received his Ph.D. from Carnegie Mellon University (CMU) in 2017 with the Best Dissertation Award. 
His research interests span different control and optimization approaches for highly dynamic robotics including nonlinear control, trajectory optimization, real-time optimization-based control, robust and adaptive control. 
\end{IEEEbiography}
\vspace{-20pt}
\begin{IEEEbiography}[{
		\includegraphics[width=1in,height=1.25in,clip,keepaspectratio]{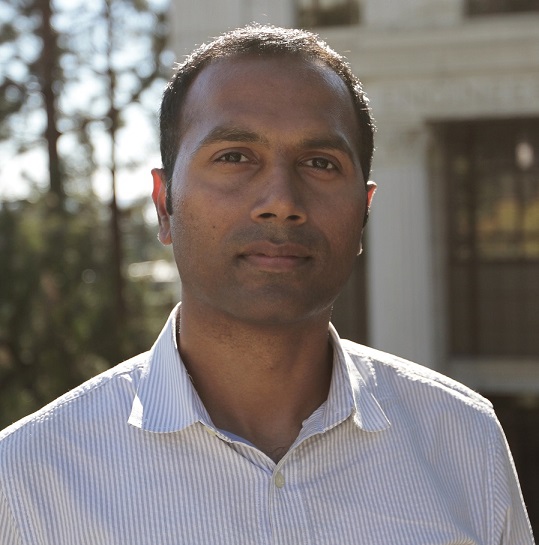}
	}]{Koushil Sreenath}
is an Assistant Professor of Mechanical Engineering at University of California, Berkeley.  He received the Ph.D. degree in Electrical Engineering: Systems and the M.S. degree in Applied Mathematics from the University of Michigan at Ann Arbor, MI, in 2011.  He was an Assistant Professor at Carnegie Mellon University between 2013-2017.  His research interest lies at the intersection of highly dynamic robotics and applied nonlinear control.  
\end{IEEEbiography}

\end{document}